\documentclass[letterpaper, 10 pt, conference]{ieeeconf}
\IEEEoverridecommandlockouts  
\usepackage{graphics,epstopdf,wrapfig}
\usepackage{epsfig}
\usepackage{subfig}
\usepackage{amsmath}
\usepackage{mathrsfs}
\usepackage{epsfig}
\usepackage{multirow}
\usepackage{amssymb,latexsym,amsfonts,amsmath}
\usepackage{graphicx}
\usepackage{color}
\usepackage{epstopdf}
\usepackage{float}
\usepackage{tikz}
\usepackage{caption}
\usepackage{verbatim}
\usepackage{tikz-3dplot}
\usetikzlibrary{shapes.geometric}
\usepackage{environ}
\NewEnviron{myequation}{%
\begin{equation}
\scalebox{1}{$\BODY$}
\end{equation}
}
\tdplotsetmaincoords{70}{30}
\usepackage{algpseudocode}
\usepackage{mathtools}
\usepackage[lined,ruled,commentsnumbered]{algorithm2e}

\makeatletter
\newcommand{\removelatexerror}{\let\@latex@error\@gobble}
\makeatother
\usepackage{amsthm}
\usepackage{pgfplots}
\pgfdeclarelayer{background}
\pgfdeclarelayer{foreground}
\pgfsetlayers{background,main,foreground}
\usetikzlibrary{intersections}
\usetikzlibrary{patterns}
\usetikzlibrary{automata,positioning,shapes,arrows}

\newtheorem{thm}{Theorem}

\DeclareMathOperator*{\argmax}{argmax} 
\usepackage{tikz}
\usetikzlibrary{automata,positioning,shapes,arrows}
\usetikzlibrary{shapes,snakes}
\tikzstyle{block} = [draw, rectangle, minimum size=3em]
\tikzstyle{bigblock} = [draw, rectangle, minimum height=7em, minimum width=11em]
\theoremstyle{remark}

\tikzset{>=latex}
\newenvironment{definition}[1][Definition]{\begin{trivlist}
		\item[\hskip \labelsep {\bfseries #1}]}{\end{trivlist}}

\theoremstyle{exampstyle}

\title{Cost-Bounded Active Classification Using\\ Partially Observable Markov Decision Processes}

\author{Bo~Wu, Mohamadreza Ahmadi, Suda Bharadwaj, and Ufuk Topcu
	\thanks{ Bo Wu, Mohamadreza Ahmadi, Suda Bharadwaj and Ufuk Topcu are with the Department of Aerospace Engineering
and Engineering Mechanics, and the Institute for Computational
Engineering and Sciences (ICES), University of Texas, Austin, 201 E 24th
St, Austin, TX 78712. email: {\tt\small $\{$bwu3, mrahmadi, suda.b, utopcu$\}$@utexas.edu}}}% <-this % stops a space

\begin{document}
\maketitle
\begin{abstract}
Active classification, i.e., the sequential decision making process aimed at data acquisition for classification purposes, arises naturally in many applications, including medical diagnosis, intrusion detection, and object tracking. In this work, we study the problem of actively classifying dynamical systems with a  finite set of  Markov decision process (MDP) models.  We are interested in finding strategies that actively interact with the dynamical system, and observe its reactions so that the true model is determined efficiently  with high confidence. To this end, we present a decision-theoretic framework based on partially observable Markov decision processes (POMDPs). The proposed framework relies on assigning a classification belief (a probability distribution) to each candidate MDP model. Given an initial belief, some misclassification probabilities, a cost bound, and a finite time horizon, we design POMDP strategies  leading to  classification decisions. We present two different approaches to find such  strategies. The first approach computes the optimal strategy ``exactly'' using value iteration. To overcome the computational complexity of finding exact solutions, the second approach  is based on adaptive sampling to approximate the optimal probability of reaching a classification decision. We illustrate the proposed methodology using two examples from medical diagnosis and intruder detection.
\end{abstract}

\section{Introduction}
%Generally speaking, the classification problem is to study how to predict which class a case or object belongs to, based on possibly noise measurement data obtained from the observation \cite{breiman2017classification}. As machine learning has been applied to an increasing number of applications, such as object recognition \cite{lowe1999object}, cancer diagnosis \cite{esteva2017dermatologist}, video classification \cite{karpathy2014large} and human behavior classification \cite{cook2015activity}, to name a few, the complexity of the classification tasks also grows significantly \cite{gao2011active}. %On the one hand, many classification problems involve a large number of features, where a subset of such features may be detected by some weaker classifiers that should be selected strategically to have the optimized classification accuracy. On the other hand, the objects in some classification tasks involve complex spatial and temporal dynamics where each sensing action may also affect the sensing outcomes in the future. Furthermore, besides sensing activities, probing activities are also available to  perturb the behaviors of the object of interest, such that a fast and reliable classification can be achieved.
We consider the following scenario as a running example. A doctor needs to implement certain tests and treatments  to determine which disease a patient has in order to prescribe further specialized tests and/or treatments. The doctor makes a diagnosis from a finite set of (known) candidate diseases, which can  evolve dynamically, e.g. the symptoms can change, with time, tests, and treatment choices.   Each test or treatment, which is implemented sequentially, can incur uncertain reactions and some costs to the patient. For each candidate disease, it is possible to 
 %Then, based on the history of tests and treatments, if the disease is already diagnosed, it is possible to 
 probabilistically predict the patient's reaction  to tests or treatments based on available historical data.  
 At each step, the doctor is required to decide whether to make a diagnosis or to implement further tests to gather more information based on the history of the patient's test results as well as the predicted reactions from the patients. On the other hand, the diagnosis procedure should not last too long, since the patient's health condition may worsen. Additionally, the cost to reach a diagnosis should be bounded to make it affordable to the patient. Such a classification procedure requires sequentially planning  of the tests, while respecting a given cost.%, which is termed as the cost-sensitive active classification problem. 

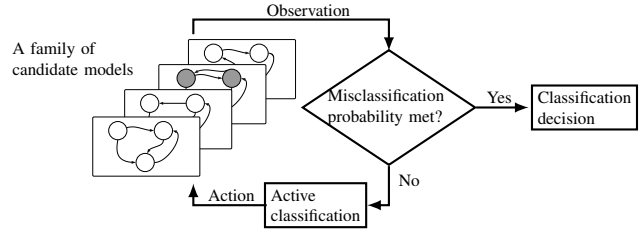
\begin{figure}[t]
	\centering

		\begin{tikzpicture}[shorten >=1pt,node distance=3cm,on grid,auto, bend angle=20, thick,scale=0.65, every node/.style={transform shape}] 
		\node[bigblock,dotted] (s0)   {}; 
		\node [draw, diamond, aspect=1.5]  (s1) [right =   4 cm of s0] {~~~~~~~~~~~~~~~~~~~~~~~~~}; 	
		\node[draw,align=left] (s2) [ right = 4 cm  of s1]  {Classification\\ decision}; 
			\node[draw,align=left] (s3) [ below left = 2 cm and 1.5 cm  of s1]  {Active\\ classification}; 
		
		%\node[state] (s3)[above left = 0.5cm and 5.5cm of s1]   {}; 
	    %\node[state] (s4) [right= 3 cm of s3] {}; 	
		%\node[state] (s5) [below left = 1 cm and 1.5cm of s4]  {}; 

        \node at (0,0){\includegraphics[scale=0.2]{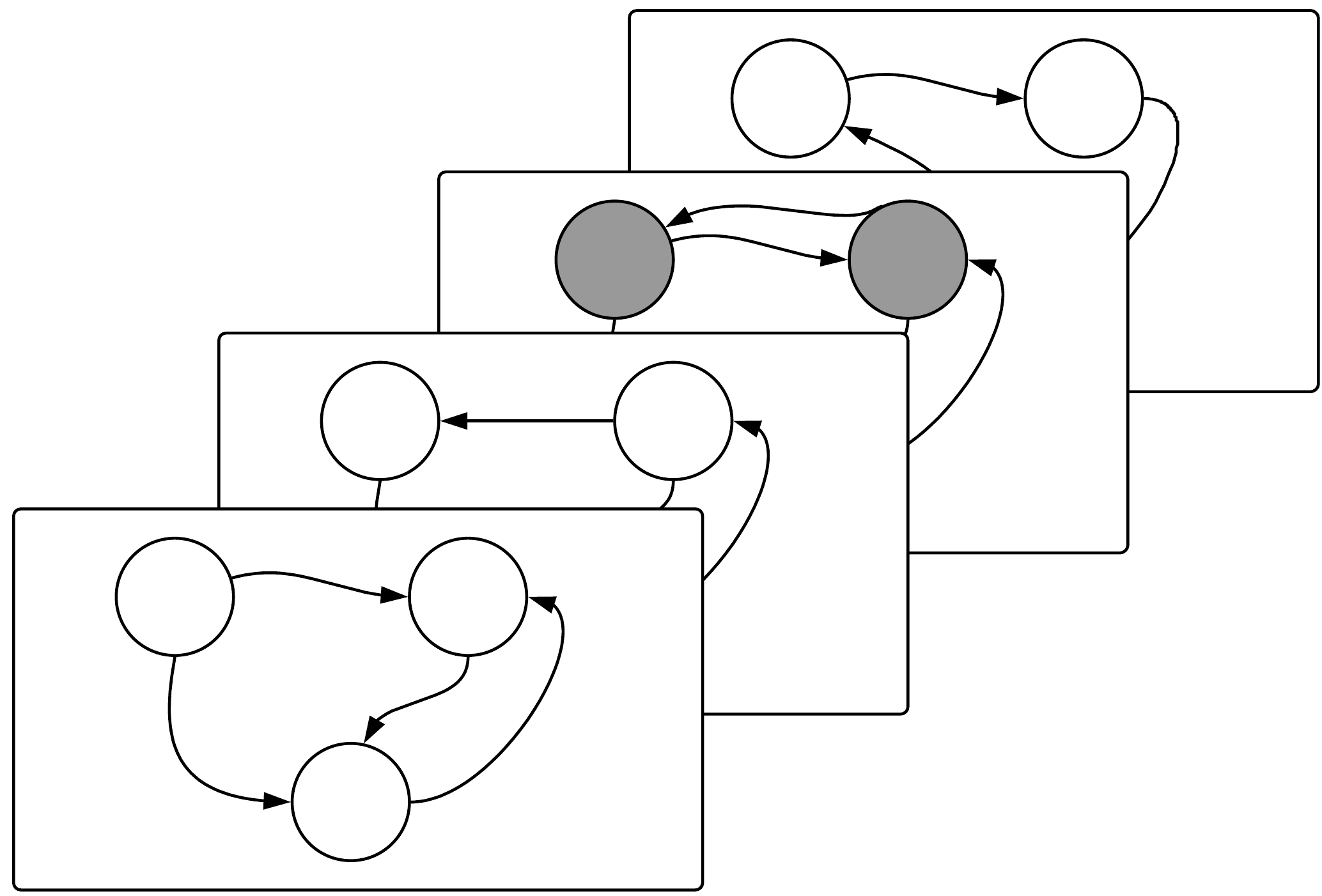}};
			
		\draw [->](0,1.5) -- (0,1.8)--(4,1.8)--(4,1.1);
		\draw [->](4,-1.2)--(4,-2)--(3.5,-2);
		\draw [->](1.5,-2)--(0,-2)--(0,-1.5);
        %\draw[->] (0,3) -| (3,1);
    	\node[text width=3cm] at (-2.2,1){A family of \\ candidate models};
    	\node[text width=3cm] at (1.8,-1.8){Action};
    	\node[text width=3cm] at (3,2){Observation};
    	\node[text width=3cm] at (7.5,0.2){Yes};
    	\node[text width=3cm] at (5.7,-1.5){No};
    	\node[text width=3cm] at (4.25,0){Misclassification \\ probability met?};
		\path[->]
				(s1) edge node [above] {} (s2)  
				%(s3) edge [bend left=30] node [above] {} (s4)  
				%(s4) edge [bend left=30] node [above] {} (s5) 	
				%(s5) edge [bend left=30] node [above] {} (s3) 	

		; %end path 
		\end{tikzpicture} 
	
	\caption{Active classification from a family of candidate models. The model with shaded states is the true one.}
	\label{fig:classification}
	\vspace{-5mm}
\end{figure}

The conventional approach for classification, which we refer to as  \emph{passive} classification, is an open-loop process that determines the class  of the object of interest based on whatever data provided. While a challenging problem itself, such an approach is not suitable for the scenario described above. Passive classification cannot decide what, when and how to obtain data to best assist the classification task. As a result, excessive amount of data (for example, the tests in diagnosis scenario) may have to be collected. Furthermore, potentially uninformative data and ignorance of the impact of the data collection process on the object of interest may make the  classification results less accurate. 

\emph{Active} classification \cite{hollinger2017active}, on the other hand, is a sequential decision-making process that has control over the data acquisition  and interacts closely with the object of interest.  Compared to passive classification, active classification is a more plausible approach in many practical classification applications, such as  medical diagnosis \cite{hauskrecht2000planning,ayer2012or}, intrusion detection, and target tracking \cite{adhikari2018applying}, where the object of interest has certain underlying dynamics belonging to a family of models. Each model is determined from the history data.  The classification process directly interacts with the state evolution of the object by selecting actions  and obtaining observations as a closed-loop system such that classification decisions can be reached more efficiently and accurately (see Figure~\ref{fig:classification}).  %

In this paper, in order to  capture the stochastic uncertainties of the outcomes associated with each action (such as a treatment) during classification, we assume the underlying dynamic model of the object to be classified belongs to a family of Markov decision processes (MDPs).   Then, we formalize an  active classification problem that can be cast into the framework of partially observable Markov decision processes (POMDPs).

As a comprehensive model for planning in partially observable stochastic environments, POMDPs characterize the uncertainties in the state evolution due to actuation imperfections or interactions with a stochastic environment as well as lack of environmental information or the observation noises. POMDPs  have been applied to a variety of applications in medical diagnosis \cite{hauskrecht2000planning,ayer2012or}, health care \cite{zois2013energy}, robotic sensing and perception \cite{hollinger2017active,myers2010POMDP,ji2007non,spaan2008cooperative} and wireless communication \cite{zhao2007decentralized}. Most existing results, however, are based on minimizing the costs incurred by actions and classification uncertainty as measured in terms of entropy~\cite{fox1998active}, where the classification accuracy is indirectly embedded in the rewards.

Because the underlying true model is not known, it is only possible to maintain a belief, which in the proposed setting, is a probability distribution over the possible MDP models and evolves  based on the history of observations and actions. The classification decision is made whenever the probability of  a certain MDP model being the true model  of the system exceeds a  given threshold based on the misclassfication probability, as shown in Figure \ref{fig:classification}.  Such a framework is particularly suitable to time- and cost-bounded classification tasks. The classification is translated to a sequential decision-making problem under uncertainties. Our objective is to find a  strategy to dynamically select classification actions such that with an optimized probability,  a classification decision can be made with a bounded cost. A similar problem was studied recently in \cite{Wang:2018:BPS:3237383.3237424}, where the belief in a POMDP must reach some goal states. However, in \cite{Wang:2018:BPS:3237383.3237424}, such requirement must be satisfied with probability $1$, which could be restrictive, considering the probabilistic  uncertainties in the belief dynamics.

Given the POMDP model and the desired misclassification probability, we propose two approaches to obtain the classification strategy that optimizes the probability to reach a classification decision. The exact solution is inspired by the cost bounded reachability in MDPs \cite{andova2003discrete}. It relies on obtaining the underlying MDP whose states are beliefs from the original POMDP. This recursive process terminates once the cost bound, a classification decision, or the maximum number of steps is reached. Then the optimal strategy can be computed from the obtained belief MDP.  To overcome the computational complexity of finding exact solutions for POMDPs, our second approach  adaptively samples the actions in the belief MDP to approximate the optimal probability to reach a classification decision. %It is then can be seen that such classification problem is a probabilistic reachability problem where the optimal strategy tries to reach the desired belief state where classification decisions can be made, at the same time, the cost can be bounded. 

%The notion of belief naturally fits in a classification problem. The classification decision can be made based on the belief over the classes to be identified where classification accuracy can be explicitly defined for each class.   % perform the  threat discrimination in dense urban terrain to detect potential armed threats, which requires the sufficient accuracy and timeliness. In general, our goal is to gather sufficient evidence to make the classification decision while minimizing the cost. We model the human behavior as a Partially Observable Markov Decision Processes (POMDP), where we can only infer if a human is hostile or not through both passive observations that are possibly noisy and active probing. We maintain a belief over the possible classes which is updated based on observations and actions performed. 
%The final classification decision can be made only when the confidence to characterize the class exceeds a threshold.  S, such as medical diagnosis and treatment or human intruder classification. In this paper, 

The rest of this paper is structured as follows. Section \ref{sec:preliminaries} provides the necessary modeling preliminaries. Section \ref{sec:Problem Formulation} formulates the problem. Two algorithms are presented in Section \ref{sec:Cost-Bounded Active Classification} to find the optimal strategy to reach a classification decision with the bounded cost. Simulation results are presented in Section \ref{sec:Simulation}. Section \ref{sec:Conclusion} concludes the paper.

%Active sensing in POMDP has been considered in robotic sensing to minimize the entropy of the belief \cite{fox1998active} or in a cooperative manner \cite{spaan2008cooperative} or localization or object detection purpose. 

\section{Preliminaries}\label{sec:preliminaries}
In this section, we describe preliminary notions and definitions used in the sequel.

\subsection{Markov Decision Processes \cite{puterman2014markov}}
 Formally, a Markov decision process (MDP)  is defined as follows.
\begin{definition}
	An MDP is a tuple $\mathcal{M}=(S,\hat{s},A,T,C)$ where
	\begin{itemize}
		\item $S$ is a finite set of states;
		\item $\hat{s}$ is the initial state;
		\item $A$ is a finite set of actions; 
		\item $T:S\times A\times S\rightarrow [0,1]$ is the probabilistic transition function with $
		T(s,a,s'):=P(s_t=s'|s_{t-1}=s,a_{t-1}=a),\forall    s,s'\in S, a\in A$; and 
		\item $C:S\times A\rightarrow\mathbb{R}_{\geq 0}$ is the cost function.
			 
	\end{itemize}
\end{definition}
In an MDP, a finite state-action path is $\omega=s_0a_0s_1a_1...$, where $s_i \in S,a_i\in A$ and $ T(s_i,a,s_{i+1})>0$.
\subsection{Hidden-Model MDPs  \cite{chades2012momdps}}\label{subsec:hmmdps}
The classification problem assumes that the object of interest is an unknown MDP that belongs to a known finite set of MDPs. Formally, it can be modeled as a hidden model MDP (HMMDP) where the   underlying true MDP is known to be one out of a finite set  $\mathcal{M}=\{\mathcal{M}_i|i\in\mathcal{C}\}$, $\mathcal{C}=\{1,...,L\}$, i.e., there are  $L$ candidate models in total.  We assume, without loss of generality, the MDPs share the same state space $S$, initial state $\hat{s}$, the action set $A$ and cost function $C$.

Given the initial state $\hat{s}$, we denote $\hat{b}_{\hat{s}}(i)$ as the initial  probability that $\mathcal{M}_i$ is the  underlying true model. Then an HMMDP is essentially a partially observable Markov decision process (POMDP) $\mathcal{P}=(Q,\pi,A,T,Z,O,C)$ where
	\begin{itemize}
        \item $Q = S\times\mathcal{C}$;
		\item $\pi:Q\rightarrow[0,1]$ is the initial state distribution with $\pi(s,i)=\hat{b}_{\hat{s}}(i),  \text{if } s=\hat{s}$ and $
		0  \text{ otherwise;}$
		\item $A$ is a finite set of actions;
		\item $T((s,i),a,(s',i'))=
		\begin{cases}
			T_i(s,a,s') & \text{ if } i=i'\\
			0  & \text{otherwise};
		\end{cases}	$ 
		\item $Z=S$ is the set of all possible observations;
		\item $O:Q\times Z\rightarrow [0,1]$ is the observation function with $O((s,i),z))=1$ if $z=s$ and $0$ otherwise;		 
		\item $C:S\times A\rightarrow\mathbb{R}_{\geq 0}$ is the cost function.
	\end{itemize}

The definition of $T$ implies that the underlying true model $\mathcal{M}_i$ will not change to any other model during the classification process. From the definition of the observation function $O$, it can be seen that the observation gives the perfect information about the state element $s$ in the state-model tuple $(s,i)$, but not the model element $i$. Therefore, when $s$ is observed, we denote $b(s,i)=b_s(i)\geq 0$ for simplicity. % as the belief that corresponds to the probability of being MDP model $i$. 

For classification purposes,  it is essential to keep track of a \emph{belief}  $b$ where $b(i)\in[0,1],\sum_{i\in\mathcal{C}}b(i)=1$, which is a probability distribution over all the possible MDP models. The belief space $B\subseteq \mathcal{R}^L$ represents the set of all possible beliefs. It is then possible to obtain a  belief MDP $\mathcal{B}=(B,\hat{b}_{\hat{s}},A,T,C_\mathcal{B})$ where
	\begin{itemize}
        \item $b\in B$ is the belief;
		\item the belief state transition probability is described by \begin{equation}\label{equation:HMMDP transition function}   T(b'_{s'},a,b_s)=\sum_i b_s(i)T_i(s,a,s'),
		\end{equation} 
		with
		\begin{equation}\label{equation:HMMDP belief update}
b'_{s'}(i)=\frac{T_i(s,a,s')b_s(i)}{\sum_j T_j(s,a,s')b_s(j)}; 	\text{ and }
		\end{equation}
		\item $C_\mathcal{B}(b_s,a)=C(s,a)$ is the cost associate to executing action $a$ at state $s$.
	\end{itemize}

\section{Problem Formulation}\label{sec:Problem Formulation}
%Assume that there is a strict horizon constraint $H$, no later than which that the classification decision has to be made. %, we are interested in computing a policy $\mu:b\rightarrow A$ to dynamically select actions based on our current beliefs over the classifications. 
To make a classification decision within a finite time bound $H$,  we keep track of the belief $b_s$ and claim the underlying model belongs to $\mathcal{M}_i$, whenever 
\begin{equation}\label{equation:reachability}
b_s(i)\geq\lambda_i,
\end{equation}
where $\lambda^i\in(0.5,1]$ denotes the minimum confidence to claim that the system belongs to $\mathcal{M}_i$. Equivalently, $1-\lambda^i$ represents the maximal acceptable error rate for $\mathcal{M}_i$. Such a belief state is a terminal state in the belief MDP. Once the terminal state is reached, the classification task is accomplished. To make the classification decision unique, i.e., there will be only one model to be declared true according to (\ref{equation:reachability}), we require that $\lambda_i>0.5,i\in\mathcal{C}$. We denote $G=\cup_{i\in\mathcal{C}} G_i$ where $G_i=\{b|b_s(i)\geq\lambda_i\}$ as the set of beliefs that should be reached to make a classification decision. %We also denote $P(\diamondsuit^{\leq H} G),$ as the probability to eventually reach the goal set $G_i$ with in $H$ steps. 

In addition to the reachability requirement, it is also essential to accomplish the classification task with a fixed amount of cost. That is, for   state-action path $\omega$ in the belief MDP $\mathcal{B}$ where $\omega=b_0a_0b_1a_1...b_N$ such that $b_i\notin G,i<H$ and $b_N\in G,N\leq H$, it is required that 
\begin{equation}\label{equation:cost}
   C(\omega)\leq D, \text{ where } C(\omega)=\sum_{i=0}^{N-1}C(a_i).
\end{equation} 
Here $\omega$ denotes a path that a classification decision is met for the first time, where $C(\omega)$ represents its accumulated cost. We denote $\Omega_G$ as the set of such paths that reach $G$ within time bounds $H$ and cost bound $D$. 

For the classification task, the objective is to compute a policy $\mu$ to dynamically select classification actions.  Since the classification should have bounded horizon and costs, the policy $\mu$ is %$\mu^*:B\times C\times I \rightarrow A$, or equivalently, 
$\mu=\{\mu_i|\mu_i:B\times E\rightarrow A,0\leq i\leq H \}$, which maps the current time step $i\in [0,H]$, the belief state $b\in B$ and  accumulated cost $e\in E=\mathbb{N}$, to an action $a\in A$.

Given a strategy $\mu$ to resolve the nondeterminism in the action selection in the HMMDP model and a path $\omega=b_0a_0b_1a_1...b_N$, it is possible to calculate its probability  
$$
P(\omega)=\prod_i T(b_i,a_i,b_{i+1}),
$$
where $\mu_i(b_i,c_i)=a_i, c_i=\sum_{j=0}^i C(a_j)$. Thus, with a strategy $\mu$, it is possible to compute the probability $P_\mu^{\leq D}(\diamondsuit^{\leq H}G)$ to reach a classification decision within time bound $H$ and cost bound $D$, where
$$
P_\mu^{\leq D}(\diamondsuit^{\leq H}G)=\sum_{\omega\in\Omega_G}P(\omega).
$$
All in all, our objective is to compute a strategy $\mu^*$  such that 
\begin{equation}{\label{equation:cost bounded reachability}}
\mu^*=\argmax_\mu P_{\mu}^{\leq D}(\diamondsuit^{\leq H}G),
\end{equation}
i.e., $\mu^*$ achieves the maximum probability to reach the decision region $G$ in $H$ steps with the accumulated cost no larger than $D$.

In many cases, reaching a classification decision may not be the only objective. For example, it is more desired to reach a diagnosis decision at the early or intermediate stages of a disease. For robotics applications, it is more desirable to classify the  object of interest without  colliding with it. Therefore, the belief may have to be constrained within a set of \emph{safe} belief $B_{safe}\subseteq B$ before the classification decision is reached. 
In such cases, the set of goal states $G$ is defined to be $G=\cup_{i\in\mathcal{C}} G_i$ where $G_i=\{b|b\in B_{safe}, b_s(i)\geq\lambda_i\}$ and the objective is to find a policy $\mu^*$ such that 
\begin{equation}{\label{equation:safe reachability}}
    \mu^*=\argmax_\mu P^{\leq D}_{\mu}(B_{safe}U^{\leq H} G),
\end{equation}
which denotes that maximized probability for the belief state to reach a classification decision $G$ while remaining in $B_{safe}$, considering the bounds on the horizon $H$ and cost $D$. It can be seen that (\ref{equation:cost bounded reachability})  is a special case of (\ref{equation:safe reachability}) where $B_{safe}=B$.

\section{Cost-Bounded Active Classification}\label{sec:Cost-Bounded Active Classification}
Given an HMMDP model with the corresponding MDPs $\mathcal{M}_i,i\in\mathcal{C}$, the time bound $H$ and cost bound $D$, in this section we introduce two approaches to solve the active classification problem as defined by (\ref{equation:cost bounded reachability}). Then we discuss how the solutions generalize to (\ref{equation:safe reachability}).

\subsection{Exact Solution}
Since the classification decision is defined on belief states as shown in (\ref{equation:reachability}),
the first step is to obtain a finite belief-state MDP $\mathcal{B}$ from the HMMDP model considering the accumulated cost. Such a procedure is called \emph{unfolding} and inspired by the similar treatment in MDPs \cite{andova2003discrete} for cost-bounded properties. In this paper, we extend this procedure to POMDPs. Without considering  the reachability and cost constraints as defined in (\ref{equation:reachability}) and (\ref{equation:cost}), the approach to obtain $\mathcal{B}$ is in Section \ref{subsec:hmmdps}. However, with cost constraints, the state space $Q$ of the belief MDP $\mathcal{B}$ is the product of the belief $B$ and the accumulated cost $E$. 

%In this section, we discuss the unfolding algorithm which promptly stops expanding a belief state when reachability constraints have been met or the cost constraints have been violated. 
The algorithm for obtaining $\mathcal{B}$ is shown in Algorithm \ref{alg:unfold}. It is essentially a recursive breadth-first traversal starting from the initial  state $\hat{q}=(\hat{b}_{\hat{s}},0)$ with initial belief $\hat{b}_{\hat{s}}$ and accumulated cost of $0$. The algorithm goes on for $H$ iterations or until there is no more state to be expanded,   as can be seen in Line (\ref{algorithm:terminate condition}). At each iteration $i$, we iterate through every state $q=(b,e)$ to be expanded (Line \ref{algorithm:q}), where $b$ is the belief state, $e$ is the cost accumulated so far. For each action $a\in A$ (Line \ref{algorithm:action} ), we calculate its next accumulated reward $e'$ d (Line \ref{algorithm:cost}). We terminate the expanding if $e'>D$ (Line \ref{algorithm:cost violation}), i.e., when the cost bound is exceeded. Otherwise,  the successor belief state $b'$ is computed (Line \ref{algorithm:belief}) as well as the transition probability (Line \ref{algorithm:transition}). If $b'\in G$,  a  classification decision is reached, otherwise  $q'=(b',e')$ is added to a set $Next$ and will be expanded in the next iteration (Line \ref{algorithm:next state}). By construction, it is not hard to see that no state in $\mathcal{B}$ will violate the cost bound and every state whose belief component belongs to $G$ is a terminal state.

\begin{figure}[!t]
		\removelatexerror
		\begin{algorithm}[H]
			\SetKwData{Left}{left}\SetKwData{This}{this}\SetKwData{Up}{up}
			\SetKwFunction{Union}{Union}\SetKwFunction{FindCompress}{FindCompress}
			\SetKwInOut{Input}{input}\SetKwInOut{Output}{output}
			\Input{  An HMMDP model with MDPs $\mathcal{M}_i,i\in\mathcal{C}$, $\mathcal{C}=\{1,...,L\}$, time bound $H$, cost found $D$, initial belief $\hat{b}_{\hat{s}}$ and reachability constraint as defined in (\ref{equation:reachability}).}
			\Output{Finite state MDP $\mathcal{B}=(Q=B\times E,\hat{b},A,T)$.}
			\BlankLine
			%\emph{special treatment of the first line}\;
			\nl $Q=\{(\hat{b}_{\hat{s}},0)\},Cur=\{(\hat{b}_{\hat{s}},0)\},i=0$;\
			
			\nl	\While {$i<H$ and $Cur\neq\emptyset$ \label{algorithm:terminate condition}}{
			        \nl $Next=\{\}$\;
    			    \nl \For{$q=(b_s,e)\in Cur$ \label{algorithm:q}}{
    				\nl \For{$a\in A$\label{algorithm:action}}{
    				\nl  $e'=e+C(s,a)$ \label{algorithm:cost}\;
    				\nl \If{$e'\leq D$ \label{algorithm:cost violation}}{
    				 \nl  \For{$s'\in S$}{
        				\nl Compute $b'$ according to (\ref{equation:HMMDP belief update}) \label{algorithm:belief}\;
        				\nl Let $q'=(b',e')$, $T(q,a,q')$ is computed according to (\ref{equation:HMMDP transition function})\label{algorithm:transition}\;
        				\nl \If {$q'\notin Q$}{
        				 \nl  $Q=Q\cup  q'$, \;
        				\nl \If {$b'\notin G$}{
        				  \nl  $Next=Next\cup q'$\label{algorithm:next state}\;
        				}}
        				}	
    				}}
    			}
					\nl $Cur=Next$\;}

			\nl \Return {$\mathcal{B}$}\;
			
			\caption{Cost-Bounded Unfolding}\label{alg:unfold}			
		\end{algorithm}
		\vspace{-3mm}
	\end{figure}

From the output of Algorithm \ref{alg:unfold}, it can be observed that the accumulated cost is already encoded in the state space of $\mathcal{B}$, therefore the cost function component in $\mathcal{B}$  is omitted. Once $\mathcal{B}$ is obtained, it is then possible to calculate the optimal strategy $\mu^*$ on $\mathcal{B}$ to achieve the following probability
\begin{equation}\label{equation:reachability1}
P_{max}(\diamondsuit^{\leq H} G)=P_{\mu^*}(\diamondsuit^{\leq H}G),
\end{equation}
i.e., the maximized probability to reach a classification decision within $H$ steps but without considering the cost bound. 

To get $\mu^*$, it is needed to compute the maximal probability, denoted as $P^q_{max}(\diamondsuit^{\leq k }G)$ to reach $G$ with in $k\in\{0,...,H\}$ steps from any $q\in Q$:
$$
P^{\hat{q}}_{max}(\diamondsuit^{\leq H }G) = P_{max}(\diamondsuit^{\leq H} G).
$$
We first divide the state set $Q$ into two disjoint subsets  $Q^{yes}=\{q=(b,e)|b\in G\}$ and $ Q^?=Q\backslash Q^{yes}$. The computation of $P^q_{max}(\diamondsuit^{\leq k }G)$ is essentially a  dynamic program as shown below. %The complexity is polynomial with respect to the number of states in this MDP \cite{rutten2004mathematical}.
\begin{align}
    P^q_{max}(\diamondsuit^{\leq i} G)&=1, ~~~\forall q\in Q^{yes},i\in\{0,...,H\},\\
    P^q_{max}(\diamondsuit^{\leq 0} G)&=0, ~~~\forall q\in Q^{?}, and \\
    P^q_{max}(\diamondsuit^{\leq i} G)&=\max_{a\in A} \sum_{q'\in Q}T(q,a,q')P^{q'}_{max}(\diamondsuit^{\leq i-1}G),\\
    &\forall q\in Q^?,i\in\{1,...,H\}\nonumber.
\end{align}
Then it can be seen that  for $i\in\{1,...,H\}$,
$$\mu^*_i(q)=\argmax_{a\in A}\sum_{q'\in Q}T(q,a,q')P^{q'}_{max}(\diamondsuit^{\leq i-1}G).$$

\subsection{Approximate Solution via Adaptive Sampling in Belief Space}

\begin{figure}[t]
		\removelatexerror
		\begin{algorithm}[H]
			\SetKwData{Left}{left}\SetKwData{This}{this}\SetKwData{Up}{up}
			\SetKwFunction{Union}{Union}\SetKwFunction{FindCompress}{FindCompress}
			\SetKwInOut{Input}{input}\SetKwInOut{Output}{output}
			\Input{  A state $q=(b_s,e)$ in $\mathcal{B}$, the number of samples $N_i$, time horizon $i$.}
			\Output{The estimated maximal probability $ \tilde{P}_i^{N_i}(q)$.}
			\BlankLine
		    \nl \If{$e>D$ or $i>H$ \label{algorithm:check1}}{
			    \Return {0}\;
			}
			\nl \If{$b_s\in G$\label{algorithm:check2}}{
			    \Return {1}\;
			}
			\nl \For{$a\in A$ \label{algorithm:init1}}{
			    \nl Sample a next state $q'=(b',e')$ by taking action $a$, where $e'=e+C(s,a)$\label{algorithm:sample1}\;
			    \nl $N_{a,i}^q=1$\;
			    \nl $\tilde{Q}(q,a)=$ CB-AMS($q',N_{i+1},i+1$) \label{algorithm:init2}\;
			}
			\nl $n=|A|$\label{algorithm:init3}\;
			\nl	\While {$n<N_i$ }{
			        \nl $a^*=\argmax_a(\frac{\tilde{Q}(q,a)}{N_{a,i}^q}+\sqrt{\frac{2\ln{n}}{N_{a,i}^q}})$\label{algorithm:action selection}\;
			        \nl Sample a next state $q'=(b',e')$ by taking action $a^*$, where $e'=e+C(s,a^*)$\label{algorithm:sample2}\;
			        \nl $\tilde{Q}(q,a^*)=\tilde{Q}(q,a^*)+$ CB-AMS($q',N_{i+1},i+1$) \label{algorithm:recur}\;
			        \nl $N_{a^*,i}^q=N_{a^*,i}^q+1$, $n=n+1$\;
		    }
            \nl $ \tilde{P}_i^{N_i}(q)=\frac{1}{N_i}\sum_a\tilde{Q}(q,a)$\;
			\nl \Return {$ \tilde{P}_i^{N_i}(q)$}\;
			
		\caption{Cost-bounded adaptive multi-stage sampling (CB-AMS)}\label{alg:ams}			
		\end{algorithm}
			\vspace{-3mm}
	\end{figure}
	
As can be seen in Algorithm~\ref{alg:unfold}, the exact solution involves constructing the belief MDP $\mathcal{B}$, whose state space grows exponentially with the time bound $H$. It leads to the well-known curse of dimensionality in the computation of $P_{max}(\diamondsuit^{\leq H}G)$ and hinders the application of Algorithm \ref{alg:unfold} to  problems with  large state spaces and a large time horizon $H$. Furthermore, the memory needed to store $\mathcal{B}$   grows exponentially as well. 

To address both the complex belief MDP $\mathcal{B}$ and intractable computation of $P_{max}(\diamondsuit^{\leq H}G)$, we now propose to use sampling algorithms \cite{chang2013simulation} to estimate the optimal classification probabilities. In particular, we leverage the adaptive multi-stage sampling algorithm (AMS) proposed in \cite{chang2005adaptive}.  The key observations for the active classification problem in (\ref{equation:cost bounded reachability}) that make the AMS a reasonable choice are as follows. First, since the  MDP models in $\mathcal{M}$ are typically  smaller than the belief MDP $\mathcal{B}$, it is easier to simulate sample paths in $\mathcal{B}$ than explicitly specifying $\mathcal{B}$ itself. Furthermore, AMS is  particularly suitable for models with a large state space but small action space \cite{chang2005adaptive}, where it is unlikely to revisit the same belief state multiple times in a sampled run, which is exactly the case the belief MDP $\mathcal{B}$ obtained with Algorithm \ref{alg:unfold}. It can be observed that in $\mathcal{B}$, the action space remains the same as  the original MDPs. Furthermore, the belief states in $\mathcal{B}$ takes values in a continuous space where generally it is very rare to revisit the same belief state with the same accumulated cost in a simulated run. Algorithm \ref{alg:ams} shows the belief state sampling procedure, termed as CB-AMS short for cost-bounded adaptive multi-stage sampling.

The input to Algorithm \ref{alg:ams} is a state $q=(b,e)$ in $\mathcal{B}$, the number of $N_i$ samples to be collected   and the current time step $i$. The output is $ \tilde{P}_i^{N_i}(q)$ which is the estimated maximal probability to reach a classification decision from state $q$ with horizon $H-i$. The initial call   is CB-AMS$((b_0,0),N_0,0)$ is for the initial belief $b_0$ and time horizon $0$. At Line \ref{algorithm:check1}, if the accumulated cost or the time horizon exceeds the bound, it will return $0$, since the probability to reach a classification decision is $0$, and there is no need to go further.  At Line \ref{algorithm:check2}, if the belief $b$ reaches its goal $G$, the algorithm will return $1$. Then from Line \ref{algorithm:init1} to Line \ref{algorithm:init2}, an initialization is performed to first try each action $a$ and sample a subsequent state $q'=(b',e')$, where $b'$ is sampled based on the transition probability as defined in (\ref{equation:HMMDP transition function}). At Line \ref{algorithm:init2}, CB-AMS is called recursively where $\tilde{Q}(q,a)$ denotes the sum of returned rewards (essentially the probabilities) by executing action $a$ from state $q$.  In Line \ref{algorithm:init3}, $n$ denotes  the number of samples collected and is initialized to be $|A|$. We will then enter the sampling loop that terminates when the number of samples $n$ reaches $N_i$. In each sampling iteration, we first select the action by the equation defined in Line \ref{algorithm:action selection}. The selection criterion balances between  sampling actions with a high average return value $\frac{\tilde{Q}(q,a)}{N_{a,i}^q}$ and trying actions that are less sampled  as denoted by $\sqrt{\frac{2\ln{n}}{N_{a,i}^q}}$, where $N_{a,i}^q$ denotes the number of times that the action $a$ has been sampled from the state $q$ at time horizon $i$. Obviously $\sum_a N_{a,i}^q=n$. It can be observed that the actions that lead to a higher probability to reach a classification decision are sampled more often and thus avoid exploring the belief states and actions that are unlikely to lead to a classification decision. 

The following theorem shows that the output of Algorithm \ref{alg:ams} converges to  $P_{max}(\diamondsuit^{\leq H} G)$  as the number of samples $N_i,0\leq i\leq H$ goes to infinity. 
\begin{thm}
If CB-AMS gets run with input $N_i$ for $i=0,...,H$ with arbitrary initial condition $q\in Q$, then
$$
\lim_{N_0\rightarrow\infty}\lim_{N_1\rightarrow\infty}\dots\lim_{N_H\rightarrow\infty}E[\tilde{P}_0^{N_0}(q)] = P^q_{max}(\diamondsuit^{\leq H}G),
$$
where $P^{q}_{max}(\diamondsuit^{\leq H}G)$ represents  the maximum probability to reach the decision region $G$ in $H$ steps with costs no larger than $D$ from state $q$.
\end{thm}
\begin{proof}
In Algorithm \ref{alg:ams}, we are effectively sampling from the MDP $\mathcal{B} $ with the state $q=(b_s,c)$ at time horizon $i$. From Line \ref{algorithm:check1} and \ref{algorithm:check2}, once $c$ exceeds the cost bound $D$ or $i$ exceeds the horizon bound $H$, no reward will be returned. Otherwise, once $b_s$ reaches the goal $G$ within the cost and horizon bound, a reward of $1$ will be returned. We denote $R(q,a)$ as the reward by executing action $a$ at state $q$. Note that this reward is not to be confused with the cost function $C$ that represents the classification cost, for example, the test and treatment costs for medical diagnosis. Equivalently, we assign $R(q,a)=1$ for all $ a\in A, q=(b_s,c),c\leq D, b_s\in G$ and $R(q,a)=0$ for $q$'s otherwise. Given a strategy $\mu=\{\mu_t|\mu_t:Q\rightarrow A,0\leq t\leq H\}$, the value function $V_i(q)$ for state $q$ and time step $i$ is
$$
V^\mu_i(q)=E[\sum_{t=i}^{H}R(q,\mu_t(q_t))|q_i=q]
$$
with $q\in Q,i=0...,H$. Given a state $q=(b_s,c)$, $V^\mu_i(q)=0$ if $c>D$ and $V^\mu_i(q)=1$, if $c\leq D$ and $b_s\in G$.  $V^\mu_{H+1}(q)=0$. Furthermore, once reaching a state $q=(b_s,c)$ with $c>D$ or $b_s\in G$, the algorithm will return the corresponding reward and such $q$ will not have successive states.  

$V^\mu_i(q)$ can be equivalently written as 
$$V^\mu_i(q)=\sum_{q'\in Q}P(q,\mu_i(q),q)V^\mu_{i+1}(q').$$
Therefore, it is not hard to see that $V^\mu_i(q)=P^\mu_{q}(\diamondsuit^{\leq H-i}G)$. Algorithm \ref{alg:ams} serves to approximate the optimal value function $V_i^*=\max_{\mu} V_i^\mu$ (equivalently the optimal probability).

Once we get the MDP with the reward structure $R$, where the value function $V_i(q)$ essentially denotes the probability to reach $G$ from $q$  within $H-i$ steps and the bounded cost, the rest of proof  follows  Theorem 3.1 in \cite{chang2005adaptive}. %Note that in its proof, it is assumed that $R_{max}\leq \frac{1}{H}$, where $R_{max}=\max_{q,a}R(q,a)$. However, such assumption is not crucial and can be relaxed as discussed in \cite{chang2005adaptive}. 
\end{proof}
Once the optimal value function (probability) has been estimated by Algorithm \ref{alg:ams}, it is then possible to extract the policy at each state $q$ and horizon $i$ by $$\mu_i(q)^*=\argmax_a \sum_{q'\in Q}P(q,a,q)\tilde{V}_i^*(q'),$$ where $\tilde{V}_i^*(q)=\tilde{P}_i^{N_i}(q)$.

Note that this sampling approach is based on the given POMDP model, which has been obtained from history data, for example, the database of medical diagnosis. Therefore,  at Line (\ref{algorithm:sample1}) and  Line (\ref{algorithm:sample2}) of Algorithm \ref{alg:ams}, we sample from known distributions as defined by the POMDP model $\mathcal{P}$, instead of trying medication actions to  patients and observe their reactions.

\subsection{Reach-Avoid in Belief Space}
Algorithm \ref{alg:unfold} and Algorithm \ref{alg:ams} can be generalized naturally for (\ref{equation:safe reachability}) as the optimization objective, where the belief  $b$ should remain in  $B_{safe}\subseteq B$. The unfolding procedure is almost identical with Algorithm \ref{alg:unfold} except in Line \ref{algorithm:next state}. There will be an additional constraint, where $b'\in B_{safe}$, i.e., if $b'$ is not in $B_{safe}$, it will not be expanded in the next iteration. For Algorithm \ref{alg:ams}, the difference is in Line \ref{algorithm:check1}, where if $b\notin B_{safe}$, it will return $0$.

%%%%%%%%%%%%%%%%%%%%%%%%%%%%%%%%%%%%%%%%%%%%%%%%%%%%%%%%%%%%%%%%%%%%%%%%%%%%%%%%%%%%%%%%%%%%%%%%%%%%%%%%%%%%%%%%%%%%%%%%%%%%%%%%%%%%%%%%%%%%%%%%%%%%%%%%%%%%%%%%%%%%%%%%%%%%%%%%%%%%%%%%%%%%%%%%%%
%%%%%%%%%%%%%%%%%%%%%%%%%%%%%%%%%%%%%%%%%%%%%%%%%%%%%%%%%%%%%%%%%%%%%%%%%%%%%%%%%%%%%%%%%%%%%%%%%%%%%%%%%%%%%%%%%%%%%%%%%%%%%%%%%%%%%%%%%%%%%%%%%%%%%%%%%%%%%%%%%%%%%%%%%%%%%%%%%%%%%%%%%%%%%%%%%%

\section{Simulation}\label{sec:Simulation}
\subsection{Medical Diagnosis and Treatment}
In this section, we introduce a particular example in medical diagnosis, where HMMDP model is used to capture how stages for a family of diseases evolve based on  tests and treatment as shown in Figure \ref{fig:hmmdp}. In particular, the states in the MDP represents the early stage, medium stage and late stage of the disease, with an increasing order of the severity. There are two possible diseases modeled by two MDPs $\mathcal{M}_1$ and $\mathcal{M}_2$ that need to be diagnosed and treated. The action space includes three actions, namely 
\begin{enumerate}
    \item $a_1$ for treatment 1,
    \item $a_2$ for treatment 2,
    \item $a_3$ for doing nothing but observe,
\end{enumerate}
where treatment $i$ as denoted by $a_i$, is more effective on the disease $i$ with costs $C(a_i)$. For $a_3$, it does not incur any cost  but the disease may have a higher probability to evolve to a later stage. The transition probabilities are as shown in the following matrices (\ref{equation:transtion probability}), where $T_i(a)(j,k)=T_i(s_j,a,s_k)$. The costs are as defined in (\ref{equation:cost at each state}) where $C(i,j)=C(s_i,a_j)$ .

		\begin{figure}[t]
			\centering
				\begin{tikzpicture}[shorten >=1pt,node distance=3cm,on grid,auto, bend angle=20, thick,scale=1, every node/.style={transform shape}] 
				\node[state,initial] (s0)   {$s_1$}; 
				\node[state] (s1) [right= of s0] {$s_2$}; 	
				\node[state] (s2) [right= of s1]  {$s_3$}; 
		        \node[text width=3cm] at (0.8,-1) {stage 1};
				\node[text width=3cm] at (3.8,-1) {stage 2};
				\node[text width=3cm] at (6.8,-1) {stage 3};
 	
				\path[->]
				(s0) edge [loop above] node {}()
				(s0) edge [bend left] node {} (s1)
				(s1) edge [loop above] node {}()
				(s1) edge [bend left] node {} (s0)  
				(s1) edge node {} (s2)  
				(s2) edge [loop above] node {}()
				; %end path 	
				\end{tikzpicture} 
			
			\caption{Feasible state transitions in HMMDP.}
\vspace{-4mm}
			\label{fig:hmmdp}
		\end{figure}
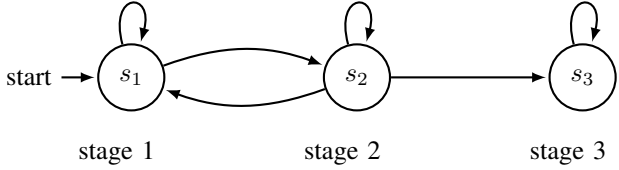

\begin{myequation}\label{equation:transtion probability}
\begin{split}
T_1(a_1)=\begin{bmatrix}
   0.8 &0.2 &0 \\
   0.7&0.2&0.1 \\
   0  &0 &1
\end{bmatrix},
T_1(a_2)=\begin{bmatrix}
   0.6 &0.4 &0 \\
   0.2 &0.4    &0.4 \\
   0 &0  &1
\end{bmatrix},\\
T_1(a_3)=\begin{bmatrix}
   0.5 &0.5 &0 \\
   0.1 &0.6    &0.3 \\
   0 &0  &1
\end{bmatrix}.
T_2(a_1)=\begin{bmatrix}
   0.6 &0.4 &0 \\
   0.1&0.5&0.4 \\
   0&0&1
\end{bmatrix},\\
T_2(a_2)=\begin{bmatrix}
   0.9 &0.1 &0 \\
   0.8 &0.1    &0.1 \\
   0 &0  &1
\end{bmatrix},
T_2(a_3)=\begin{bmatrix}
   0.3 &0.7 &0 \\
   0.1 &0.3   &0.6 \\
   0 &0  &1
\end{bmatrix},
\end{split}
\end{myequation}

\begin{myequation}\label{equation:cost at each state}
C=\begin{bmatrix}
   2 &5&0 \\
   6 &4&0 \\
   7 &7 &0
\end{bmatrix}
\end{myequation}

The diagnosis decision is made for disease $1$ or $2$ if one of following  is satisfied. 
\begin{equation}\label{equation:diagnosis requirement}
    b_s(1)\geq \lambda_1\text{ or } b_s(2)\geq \lambda_2,
\end{equation}
with the initial belief $\hat{b}_{s_1}=(0.5,0.5)$, cost constraint $D=10$. One step unfolding according to the Algorithm \ref{alg:unfold} is shown in Figure \ref{fig:unfolding 1}. If $\lambda_1=0.8,\lambda_2=0.7$, it can be seen that if $a_2$ is executed, there is  $0.25$ probability that the disease is diagnosed to be type $1$ (since $b_{s_2}(1)=0.8$) at the shaded state $q_4$, with a cost of $5$. Therefore, $q_4$ will not be included in the states to be expanded in the next iteration. 

	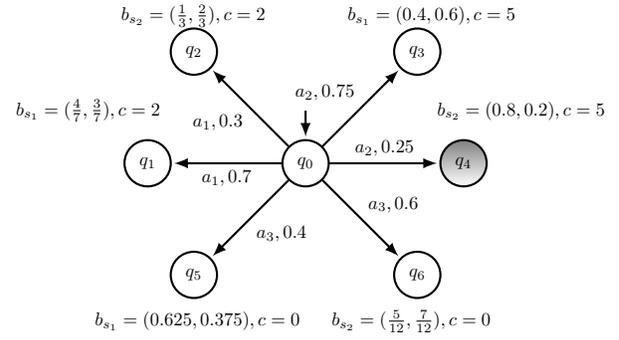
\begin{figure}
		\centering	
\begin{tikzpicture}[shorten >=1pt,node distance=3cm,on grid,auto, bend angle=20, thick,scale=0.7, every node/.style={transform shape}] 
				\node[state] (s0)   {$q_0$}; 
				\node[state] (s1) [left =of s0] {$q_1$}; 
				\node[state] (s2) [above left  = of s0]  {$q_2$}; 
				\node[state] (s3) [above right= of s0] {$q_3$}; 
				\node[state,shade] (s4) [right= of s0] {$q_4$}; 
				\node[state] (s5) [below left= of s0] {$q_5$}; 
				\node[state] (s6) [below right= of s0] {$q_6$}; 
                \node[text width=3cm] at (-4,1) {$b_{s_1}=(\frac{4}{7},\frac{3}{7}),c=2$};		
                \node[text width=3cm] at (-2,2.8) {$b_{s_2}=(\frac{1}{3},\frac{2}{3}),c=2$};
                \node[text width=4cm] at (2.8,2.8) {$b_{s_1}=(0.4,0.6),c=5$};
                \node[text width=4cm] at (4.5,1) {$b_{s_2}=(0.8,0.2),c=5$};
                \node[text width=4cm] at (-2,-3) {$b_{s_1}=(0.625,0.375),c=0$};;
                \node[text width=4cm] at (2.5,-3) {$b_{s_2}=(\frac{5}{12},\frac{7}{12}),c=0$};;
                \draw [->] (0,1) -- (s0);
				\path[->]
				
				(s0) edge node {$a_1,0.7$} (s1) 
				(s0) edge node {$a_1,0.3$} (s2) 
				(s0) edge node {$a_2,0.75$} (s3) 
				(s0) edge node {$a_2,0.25$} (s4) 
				(s0) edge node {$a_3,0.4$} (s5) 
				(s0) edge node {$a_3,0.6$} (s6) 
				; %end path 	

				\end{tikzpicture} 
		\caption{One step unfolding}\label{fig:unfolding 1}
	\end{figure}
	
We implement both Algorithm \ref{alg:unfold} and Algorithm \ref{alg:ams} in C++. For Algorithm \ref{alg:unfold}, the resulting MDP model is input into the PRISM model checker \cite{kwiatkowska2011prism} to compute the maximum probability (\ref{equation:reachability1}). For Algorithm \ref{alg:ams}, we set $N_i=2000,0\leq i\leq H$. We also store the calculated values of $\tilde{P}_i^{N_i}(q)$ to avoid recomputing them. The results are as shown in Figure \ref{fig:prob} to illustrate how the maximum probability to diagnose the disease increases with the horizon $H$.  % where $\lambda_a$ denotes $\{\lambda_1=0.8,\lambda_2=0.7\}$,$\lambda_b$ denotes $\{\lambda_1=0.9,\lambda_2=0.8\}$,$\lambda_c$ denotes $\{\lambda_1=0.95,\lambda_2=0.9\} $. 
Typically, the higher the classification threshold, the less likely that  belief states to make a classification decision can be encountered and thus there is less probability of  successful diagnosis. Furthermore, the CB-AMS algorithm is able to closely estimate the optimal probability. 

\begin{figure}
    \centering
    \includegraphics[scale=0.18]{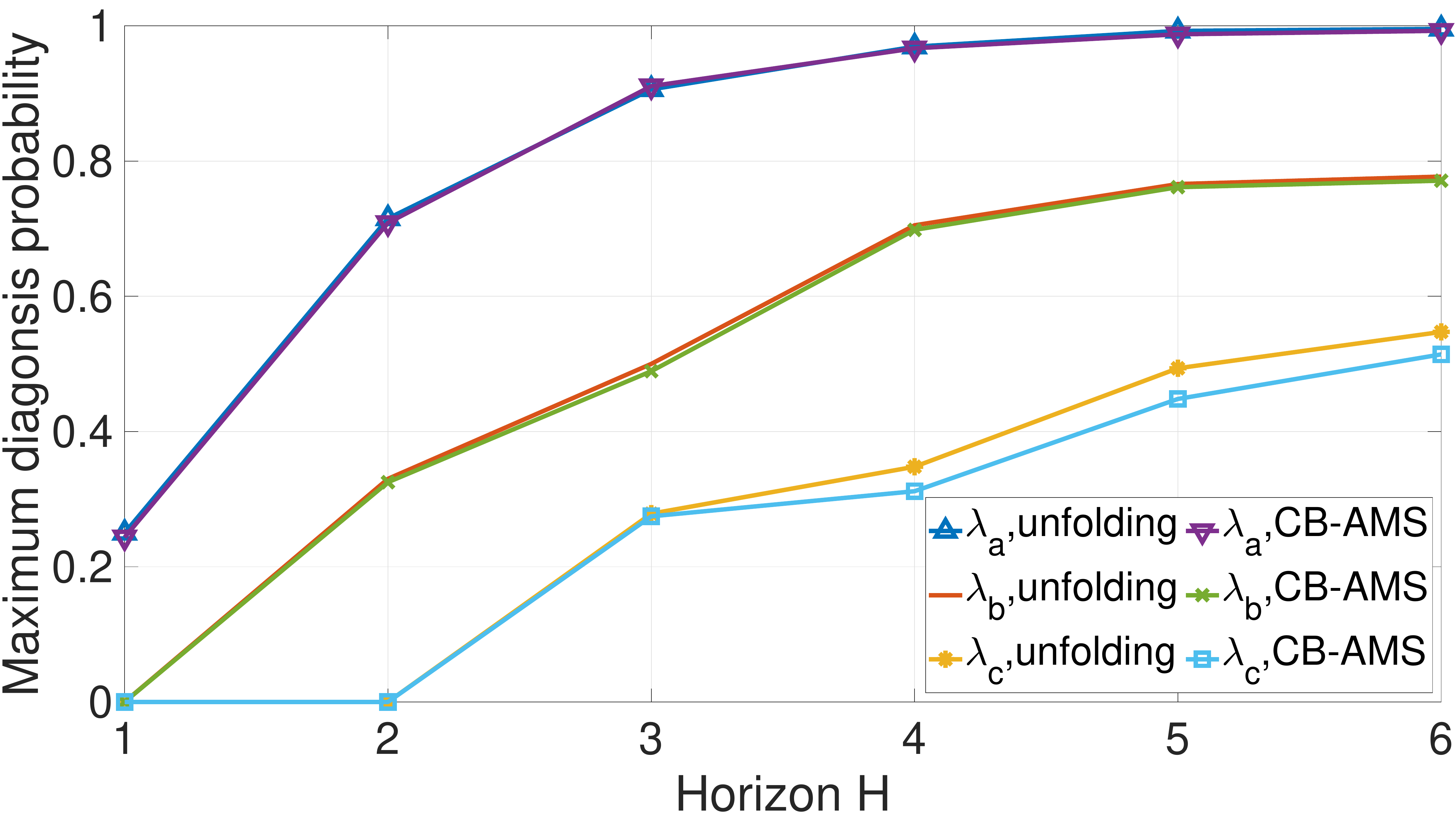}
    \caption{Maximum probability  for (\ref{equation:reachability1}).  $\lambda_a:\{\lambda_1=0.8,\lambda_2=0.7\}$, $\lambda_b:\{\lambda_1=0.9,\lambda_2=0.8\}$, $\lambda_c:\{\lambda_1=0.95,\lambda_2=0.9\}$.}
    \label{fig:prob}
    \vspace{-3mm}
\end{figure}

\begin{figure}
    \centering
    \includegraphics[scale=0.18]{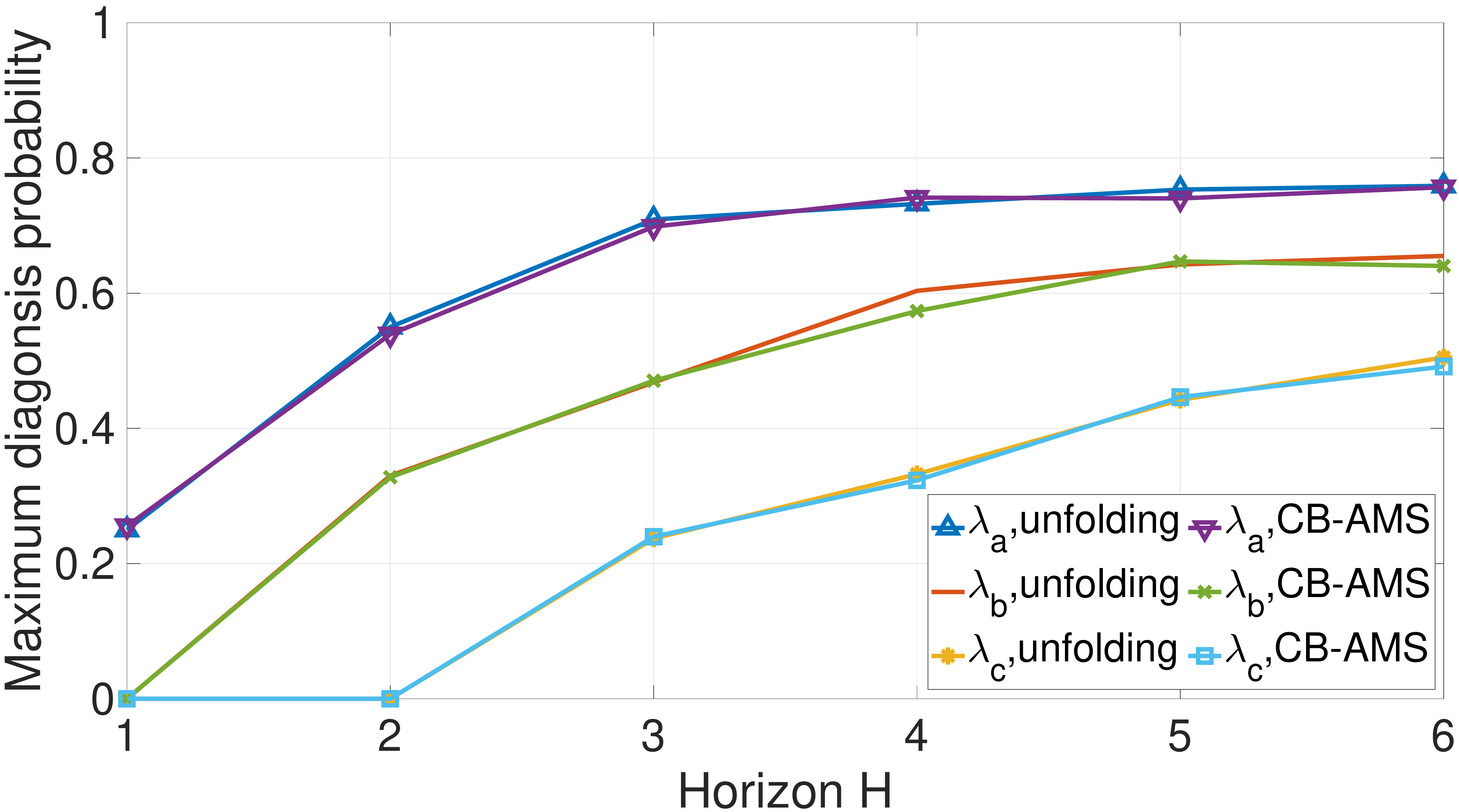}
    \caption{Maximum probability  for (\ref{equation:safe reachability}) with $B_{safe}$ defined in (\ref{equation:bsafe}). $\lambda_a:\{\lambda_1=0.8,\lambda_2=0.7\}$, $\lambda_b:\{\lambda_1=0.9,\lambda_2=0.8\}$, $\lambda_c:\{\lambda_1=0.95,\lambda_2=0.9\}$.}
    \label{fig:prob1}
\end{figure}
Now let's look at  (\ref{equation:safe reachability}) where we try to reach a classification decision without visiting undesired belief states. For the medical diagnosis example, the desired property is to diagnose the disease without reaching the late stage of the disease, where the corresponding $B_{safe}$ can be defined as
\begin{equation}{\label{equation:bsafe}}
    B_{safe}=\{b_s|s\neq s_3\}. 
\end{equation}
For the medical diagnosis example, the results are as shown in Figure \ref{fig:prob1}. The same trend as in Figure \ref{fig:prob} can be observed where the maximum probability to safely reach $G$ without going out of $B_{safe}$ increases with a longer time horizon. However, due to the extra safety constraint, the maximum probabilities also decreases, compared to the results without safety constraint. Again, the CB-AMS algorithm performs well to estimate the optimal probability. 

The run times for both algorithms with regards to specification (\ref{equation:reachability1}) and (\ref{equation:safe reachability}) are as shown in Table \ref{table:1} and Table \ref{table:2}. For Algorithm \ref{alg:unfold}, the run time consists of the time to get the belief MDP model $\mathcal{B}$ and compute the optimal probability. All the experiments were run on a laptop with 2.6GHz i7 Intel\textsuperscript{\textregistered} processor with 16GB memory.  It can be seen that for a small time horizon $H$, exact solution outperforms sampling in time consumption, as the number of the states in $\mathcal{B}$ is small. But as the horizon $H$ grows, the run time from sampling grows slower than the exact horizon and thus is more favorable.
\begin{table}[t]
	\centering
	\caption{Run times for (\ref{equation:reachability1}) in seconds}
	\begin{tabular}{|l|l|l|l|l|l|l|l|}
		\hline
		\multicolumn{2}{|l|}{Horizon $H$} & $1$ & $2$ & $3$ & $4$ & $5$ & $6$ \\ 
		\hline
	    \multirow{2}{*}{$\lambda_a$} & Unfold &0.31 & 0.93 & 3.13 & 9.88 & 24.24 & 224.12  \\\cline{2-8}
		
		& CB-AMS& 0.38 & 1.54 & 4.84 & 17.01 & 50.74 & 133.90   \\
		\hline
		\multirow{2}{*}{$\lambda_b$} & Unfold & 0.01 &	1.00 &	4.74 &	21.14 &	131.06 &	930.42 \\\cline{2-8}
		
		 & CB-AMS & 0.49 &3.83& 5.62 &25.52& 78.25 &196.18 \\
		\hline
		\multirow{2}{*}{$\lambda_c$} & Unfold & 0.01 &	0.06 &	4.48 &	20.24 &	189.67 &	1418.75 \\\cline{2-8}
		
		 & CB-AMS & 0.40 & 1.71 & 5.54 & 26.40 & 103.84 & 272.88  \\
		\hline
	\end{tabular}
 \vspace{-3mm}
	\label{table:1}
\end{table}

\begin{table}[h]
	\centering
	\caption{Run times for (\ref{equation:safe reachability}) in seconds}
	\begin{tabular}{|l|l|l|l|l|l|l|l|}
		\hline
		\multicolumn{2}{|l|}{Horizon $H$} & $1$ & $2$ & $3$ & $4$ & $5$ & $6$ \\ 
		\hline
	    \multirow{2}{*}{$\lambda_a$} & Unfold & 0.40 & 0.71 & 2.39 & 7.62 &30.34 & 229.89 \\\cline{2-8}
		
		& CB-AMS& 0.37& 1.38 & 4.26 & 14.39 & 39.48 & 98.97 \\
		\hline
		\multirow{2}{*}{$\lambda_b$} & Unfold & 0.01 &	0.92 &	3.04 &	12.31 &	94.13 &	1054.71 \\\cline{2-8}
		
		 & CB-AMS & 0.37 & 1.40& 6.11 & 21.67 & 73.12 & 235.12  \\
		\hline
		\multirow{2}{*}{$\lambda_c$} & Unfold & 0.01& 0.12 & 4.35 & 17.76 & 180.36 & 2211.70 \\\cline{2-8}
		
		 & CB-AMS & 0.39 & 1.72 & 6.20 & 28.87 & 91.46 & 330.35  \\
		\hline
	\end{tabular}
    \vspace{-3mm}
	\label{table:2}
\end{table}

\subsection{Intruder Classification}
%Active classification can be crucial for security applications. Specifically, i
In automated surveillance applications \cite{bharadwaj2017synthesis}, it is often necessary to determine whether a detected target is a potential threat before deploying security resources for further intervention. %For example, the work in \cite{bharadwaj2017synthesis} synthesizes controllers for mobile sensors with quantitative surveillance requirements.  %However, it is often the case that detected targets are not threats, such as when small animals or disturbances set off alarms. In these cases, we want to monitor the situation and determine whether a detected target is a threat before we deploy the mobile sensors to perform active surveillance. We assume we can always passively observe the target's location, which is generally possible through radar or some other static sensors in the environment,

\begin{figure}
\vspace{-2.8mm}
\centering
\subfloat{
\includegraphics[scale=0.25]{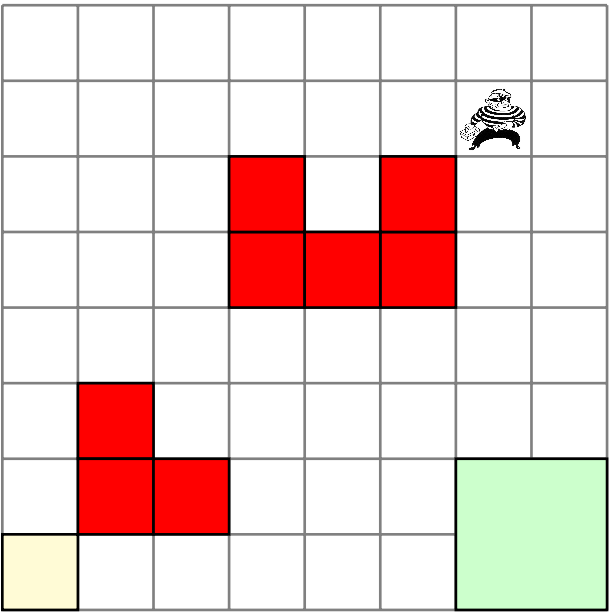}}
%\hspace{1.5cm}}
%\hfill
\subfloat{
\includegraphics[scale=0.25]{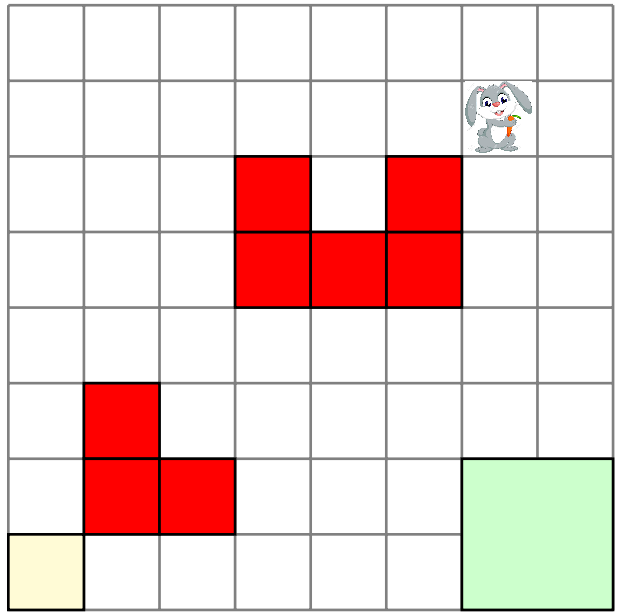}}
		\caption{Gridworld with two potential classes of targets - hostile (left) and safe (right). Green cells are sensitive areas. Red cells are obstacles such as buildings. The yellow cell is a hiding place for the hostile intruder. } \label{fig:grid}
		\vspace{-2.8mm}
\end{figure}
We present an example of a $8\times8$ gridworld shown in Figure~\ref{fig:grid}. The target is not allowed to reach the green zone. The target is assumed to be either  a hostile human intruder (Class 1) or an animal (Class 2) which has no real threat. The behaviors of the hostile and safe intruders are characterized by two MDPs $\mathcal{M}_1$ and $\mathcal{M}_2$, respectively. The state  in the MDPs refers to target's location in the gridworld, which is  observable through radar or some other static sensors in the environment.  At each time step the target moves to one of its neighbouring cells randomly.  The actions available to the automated surveillance are 
\begin{enumerate}
    \item $a_1$ for passive observation,
    \item $a_2$ for alarm through  loudspeakers.
\end{enumerate}
If $a_1$ is chosen, a human intruder will attempt to move to the sensitive area (the green region) as denoted by $S_{green}$. When $a_2$ is executed, the animal will be startled and  move in all directions with equal probability, while the human intruder will tend to move towards the yellow region  in Figure \ref{fig:grid} ostensibly to hide. The human moves randomly  but generally heads to  $S_{green}$ or yellow region for action $a_1$ or $a_2$. The randomness is to capture different human preferences and  human's inherent decision uncertainty. The costs for $a_1$ and $a_2$ at each state are $1$ and $3$, respectively.

The corresponding $B_{safe}$ can be defined as
$$
    B_{safe}=\{b_s|s\notin S_{green}\}. 
$$

The classification decision is  made  if one of following  is satisfied. 
$$
    b_s(1)\geq 0.7\text{ or } b_s(2)\geq 0.7,
$$
with the initial belief $\hat{b}_{\hat{s}}=(0.5,0.5)$, step bound  $H=6$ and cost bound $D=8$, where $\hat{s}$ is the initial state that intruders at as seen in Figure \ref{fig:grid}.
%We aim to classify a target as either hostile or safe. Moving forward, we refer to the hostile target as class 1 and the safe target as class 2.

%At each time step the target can move to one if its neighbouring cells. We assume we have MDP models of a hostile target and a non-hostile target. These models capture our knowledge of the behaviour of different classes of targets. Specifically, we assume that under no interference, both classes of targets will move towards the green region. 

%However, we insert a degree of randomness in the behaviour of each class to account for possible uncertainty in our models. For example, if the MDP model of class 1 assumes that the target will take action South from the current state, we assume there is a small probability of the target taking one of the actions North, East or West instead, parametrized by  a randomness parameter $0 \leq \gamma_1 \leq 1$. If $\gamma = 0$, there is no uncertainty and and if $\gamma = 1$, we know nothing about the target's behaviour which means it can take all actions with equal probability.

%We have the option of taking an action - sending an announcement through a loudspeaker. The MDP model for target class 1 states the target will move towards the yellow region in Figure \ref{fig:grid} (again with randomness injected) ostensibly to hide. The MDP model for target class 2 will move in all directions with equal probability. This action will make classifying easier, however, taking this action incurs more cost.

\begin{figure*}[ht]
	\begin{minipage}{5.0cm}
		\subfloat[$t=0$ \label{fig:case1t2}]{
			\includegraphics[scale=0.13]{grid_class1.png}\hspace{-0.2cm}
		}
		\subfloat[$t=3$ \label{fig:case1t3}]{
			\includegraphics[scale=0.13]{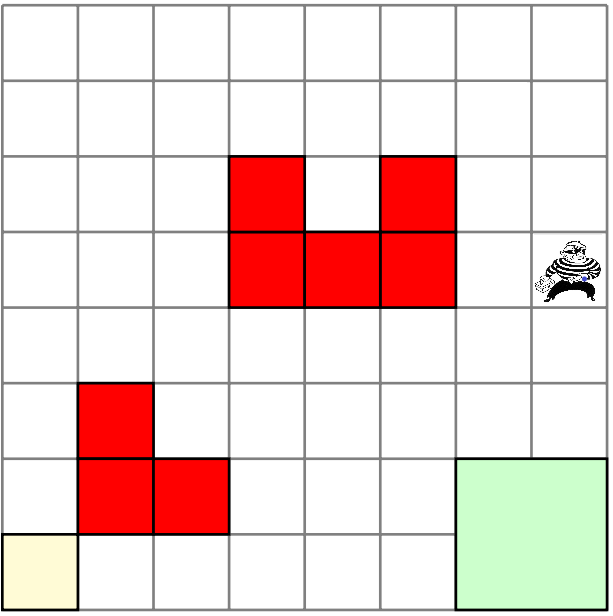}\hspace{-0.2cm}
		}
		\subfloat[$t=5$ \label{fig:case1t4}]{
			\includegraphics[scale=0.1275]{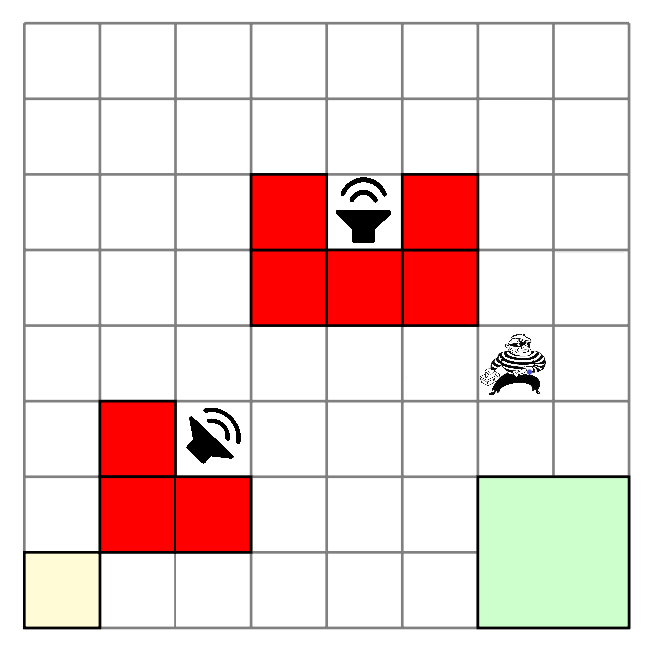}\hspace{-0.2cm}
		}
		\subfloat[$t=6$ \label{fig:case1t9}]{
			\includegraphics[scale=0.1275]{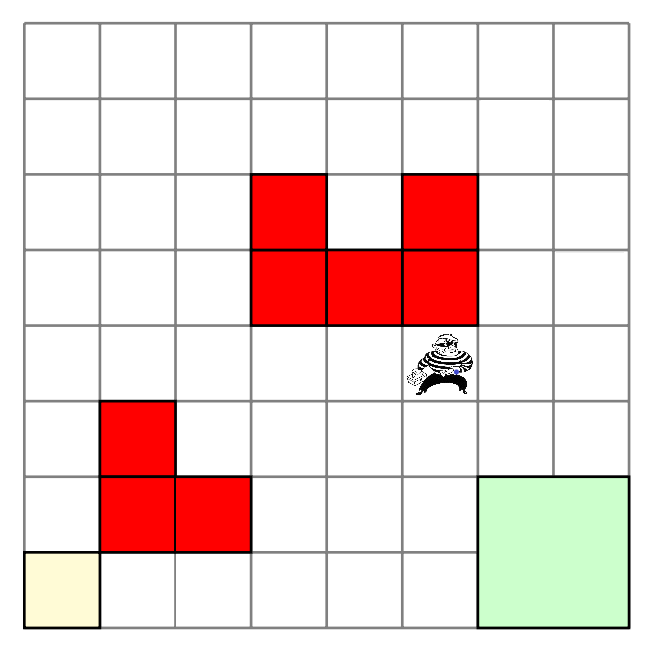}\hspace{-0.2cm}
		}
		\subfloat[$t=0$ \label{fig:case3t2}]{
			\includegraphics[scale=0.13]{grid_class2.png}\hspace{-0.2cm}
		}
		\subfloat[$t=3$ \label{fig:case2t3}]{
			\includegraphics[scale=0.13]{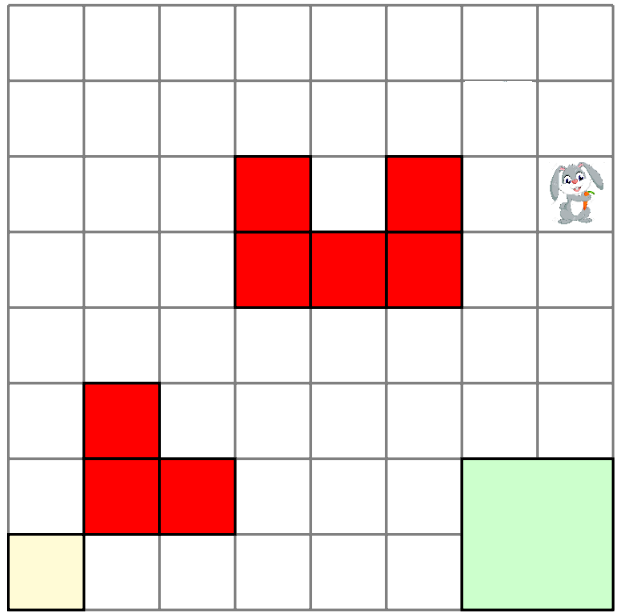}\hspace{-0.2cm}
		}
		\subfloat[$t=5$ \label{fig:case2t4}]{
			\includegraphics[scale=0.1275]{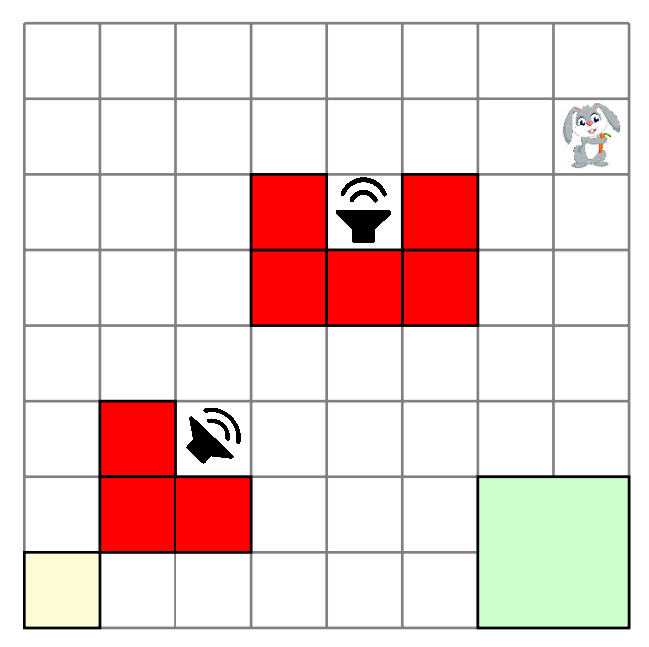}\hspace{-0.2cm}
		}
		\subfloat[$t=6$ \label{fig:case2t9}]{
			\includegraphics[scale=0.1275]{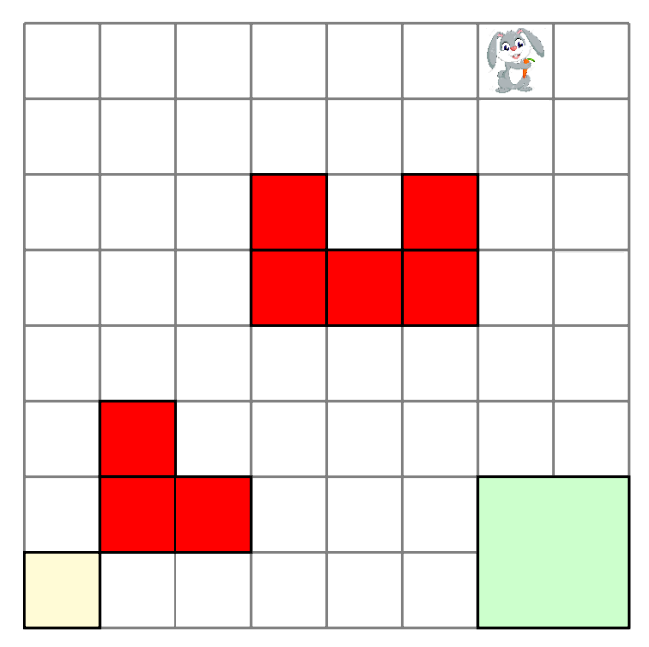}
		}
	\end{minipage}
	
	\begin{minipage}{5.0cm}
	%	\centering
		\hspace{-0.23cm}\subfloat[$t=0$  \label{fig:case1t5}]{
			\includegraphics[scale=0.095]{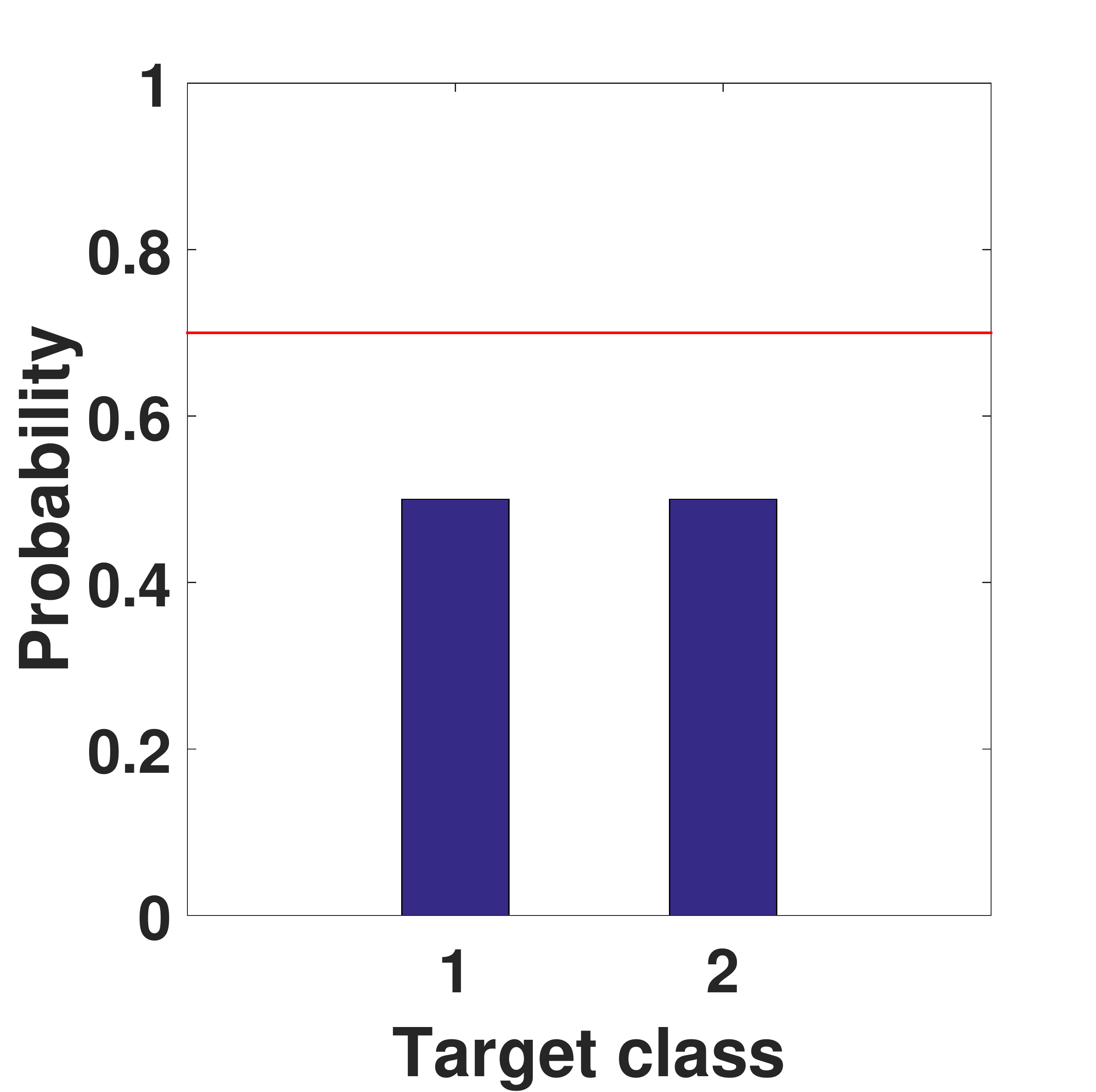}\hspace{-0.45cm}
		}
		\subfloat[$t=3$ \label{fig:case1t6}]{
			\includegraphics[scale=0.095]{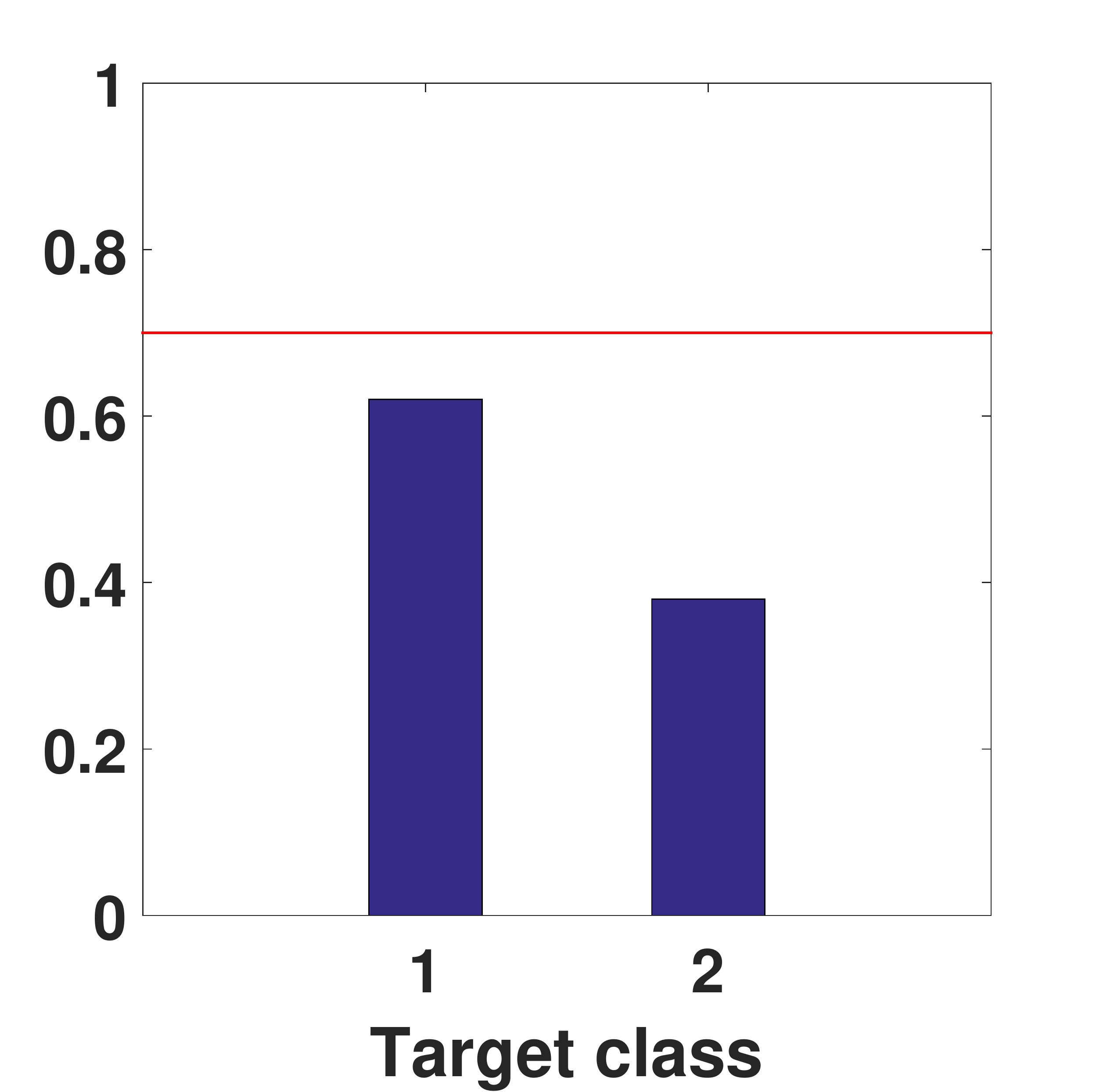}\hspace{-0.45cm}
		}
		\subfloat[$t=5$ \label{fig:case1t7}]{
			\includegraphics[scale=0.095]{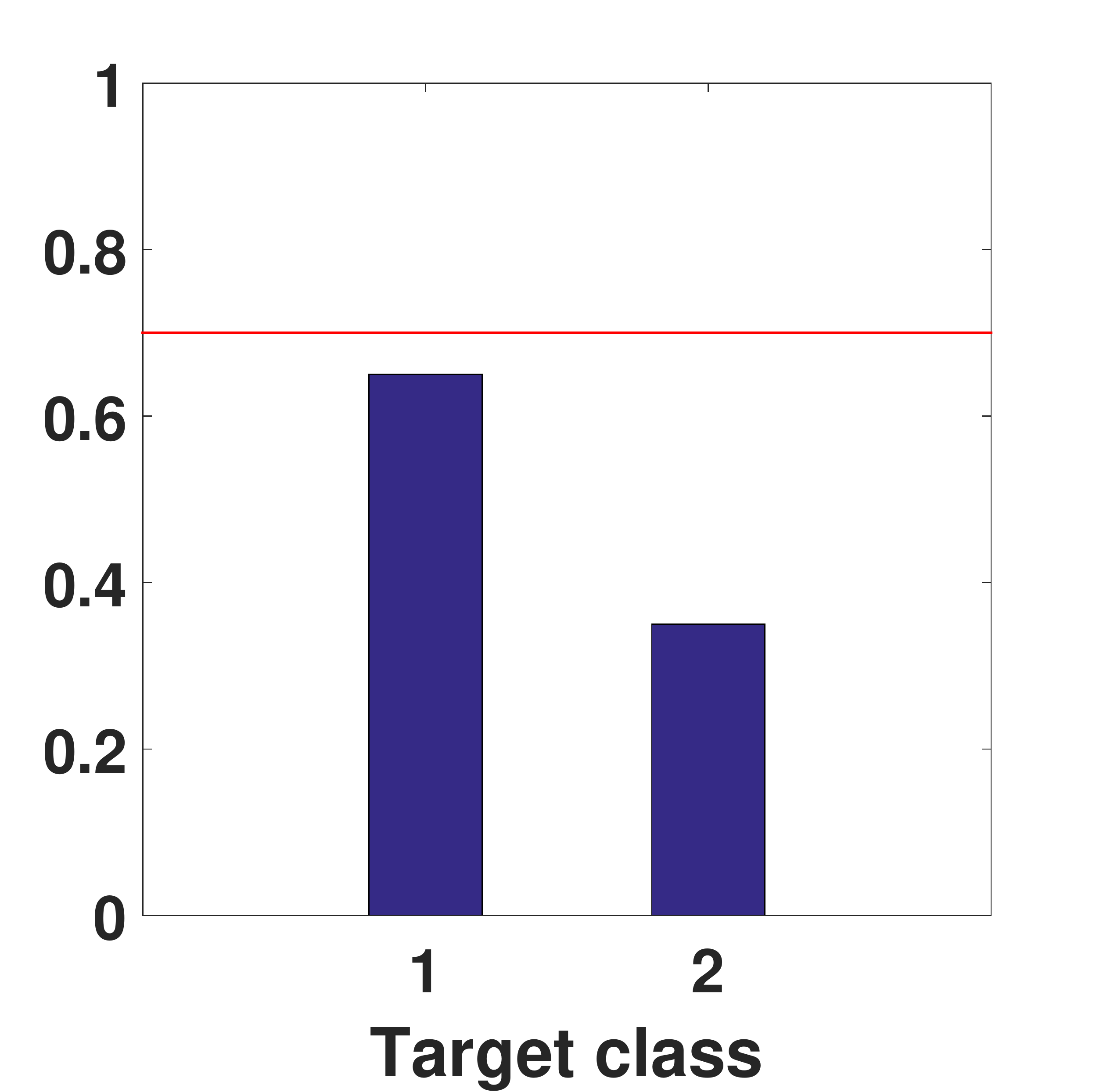}\hspace{-0.45cm}
		}
		\subfloat[$t=6$ \label{fig:case1t8}]{
			\includegraphics[scale=0.095]{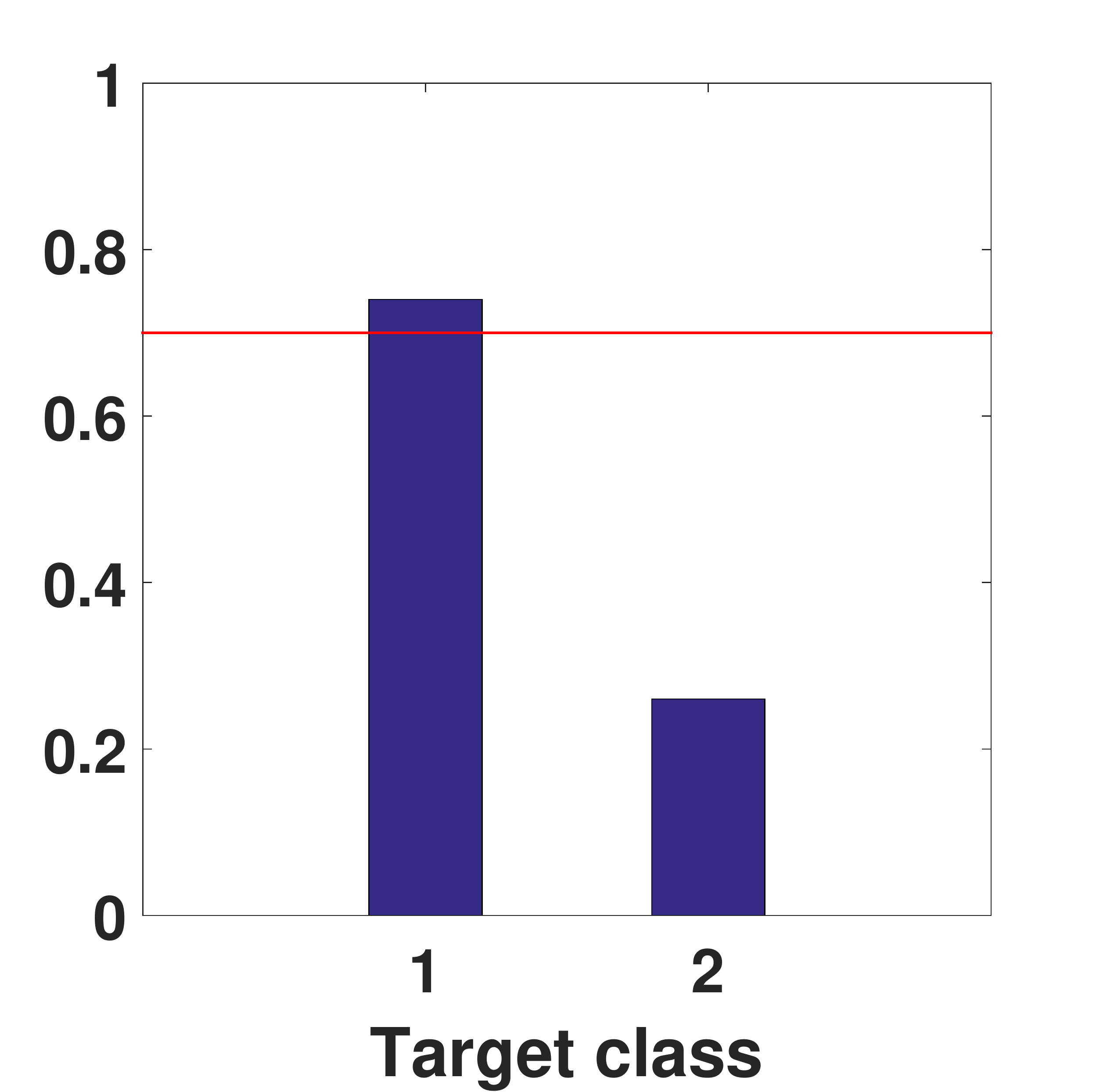}\hspace{-0.45cm}
		}
	\subfloat[$t=0$  \label{fig:case2t5}]{
			\includegraphics[scale=0.095]{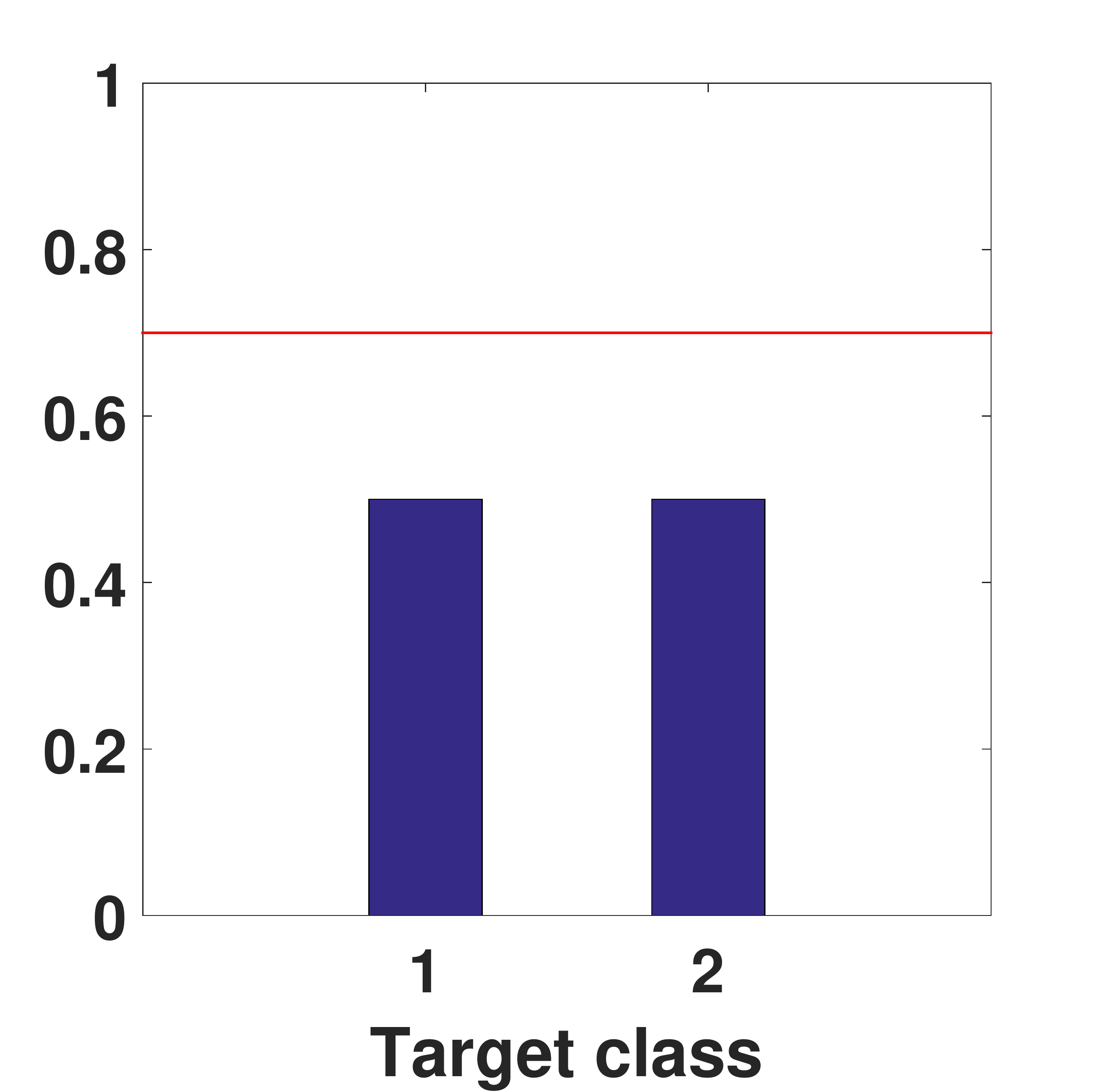}\hspace{-0.45cm}
		}
		\subfloat[$t=3$ \label{fig:case2t6}]{
			\includegraphics[scale=0.095]{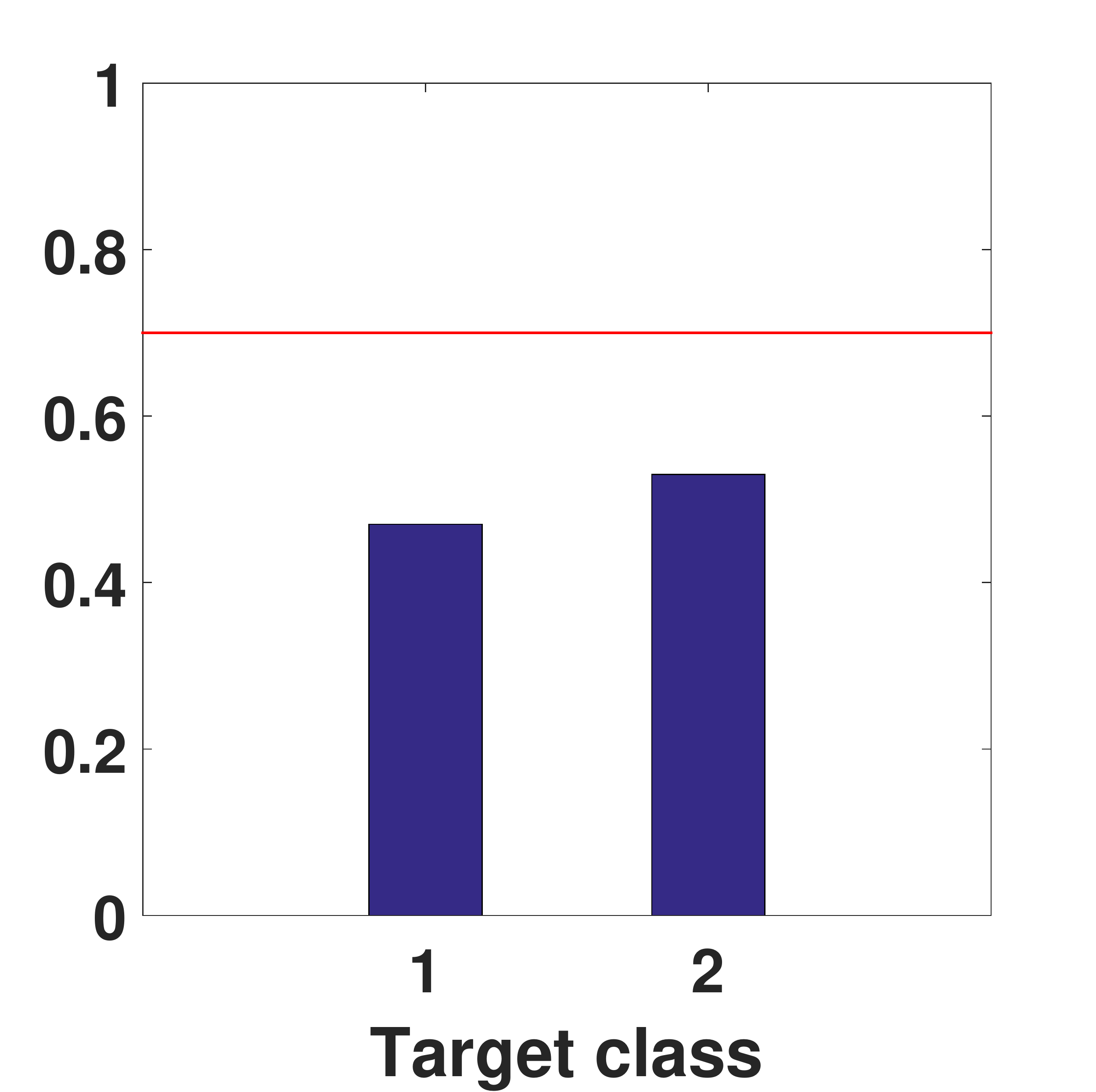}\hspace{-0.45cm}
		}
		\subfloat[$t=5$ \label{fig:case2t7}]{
			\includegraphics[scale=0.095]{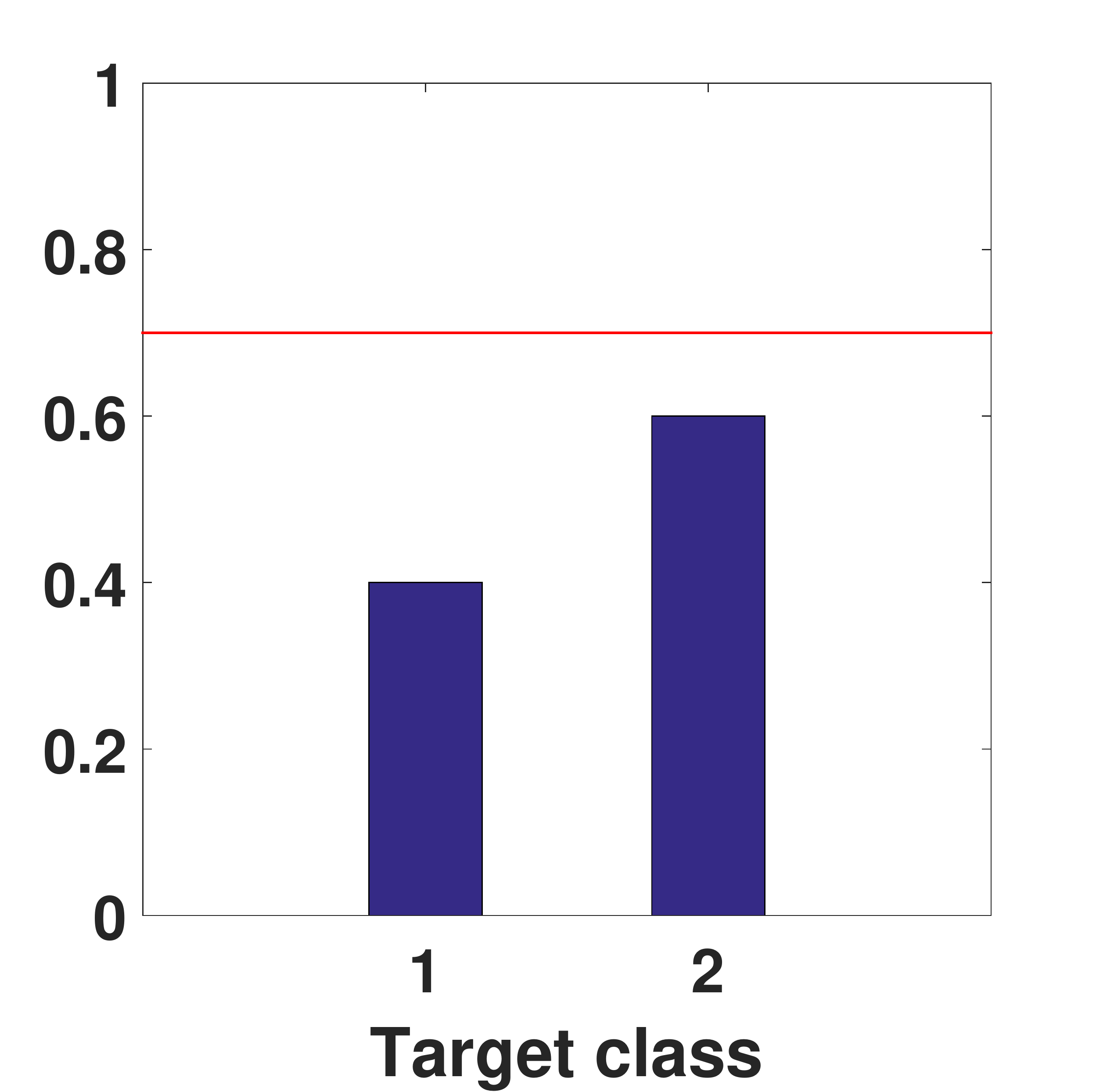}\hspace{-0.45cm}
		}
		\subfloat[$t=6$ \label{fig:case2t8}]{
			\includegraphics[scale=0.095]{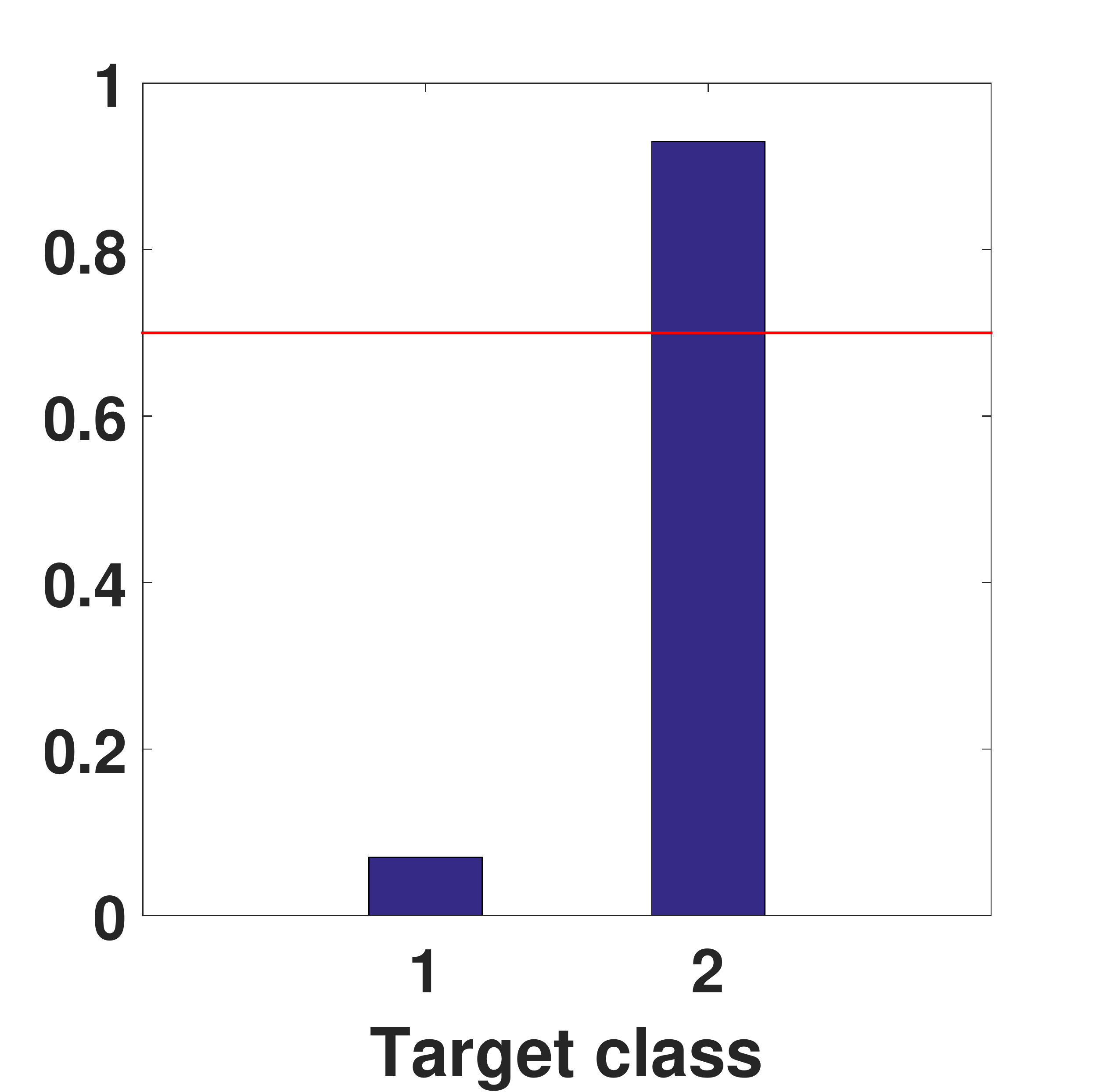}
		}
	\end{minipage}

	\caption{Simulation of the classification problem in a gridworld environment with both target classes. Figures \ref{fig:case1t2}- \ref{fig:case1t9} and \ref{fig:case1t5}-\ref{fig:case1t8} correspond to a simulation with a hostile target. Figures \ref{fig:case3t2}- \ref{fig:case2t9} and \ref{fig:case2t5}-\ref{fig:case2t8} correspond to a simulation with a safe target.}
	\label{fig:gridresults}
	\vspace{-3mm}
\end{figure*}

The results of the experiment can be seen in Figure~\ref{fig:gridresults}, where exact solution is used to compute the optimal classification strategy. Suppose the target is a hostile intruder, Figures~\ref{fig:case1t2}- \ref{fig:case1t9} illustrate a run of the human movement, where action $a_1$ is executed at $t=0,3,6$ and  action $a_2$ is executed at $t=5$. Figure~\ref{fig:case1t5}-\ref{fig:case1t8} depicts to the corresponding beliefs. 

In the scenario shown in Figure \ref{fig:case1t4}, the target is near $S_{safe}$ and the corresponding belief is  as seen in Figure~\ref{fig:case1t7} which favors class $1$. The optimal action at the time instance $t=5$ is to sound the alarm. Then at $t=6$, it is observed that the target moves towards the yellow cell and the belief of class $1$ exceeds the threshold. This is where the classification terminates and a human operator will be alerted. Figures~\ref{fig:case3t2}- \ref{fig:case2t9} and \ref{fig:case2t5}-\ref{fig:case2t8} shows the classification with a safe target (an animal) where similar behavior can be observed. After the alarm is used in $t=5$ in Figure \ref{fig:case2t7}, the rabbit runs to the top right corner as seen in Figure \ref{fig:case2t8}, a very unlikely move for a hostile intruder. As a result, the belief for class $2$ exceeds the threshold and the classification process terminates. In both simulations, the final costs are $7$.   %In some cases, it is possible to classify with only passive observation. This usually occurs when $\gamma$ is small which means our confidence in the models is high. However, as $\gamma$ is increased, more observations are needed and this may not be possible if we need to classify the target before it reaches the green region. 

\section{Conclusion}\label{sec:Conclusion}
In this paper, we studied a cost-bounded active classification of dynamical systems belonging to a  finite set of MDPs. We utilized the POMDP modeling framework and the objective was to actively select actions based on the current belief, accumulated cost, and time step, such that the probability to reach a classification decision within a cost bound can be maximized. To solve the problem, we proposed two approaches. The first one was an exact solver to obtain the unfolded belief MDP model considering the cost-bound, and then compute the optimal strategy. To mitigate the computation burden, the second approach adaptively samples the actions to estimate the maximum probability. One bottleneck of the propose approaches is the computational complexity. In the future,  we will explore point-based methods \cite{Kurniawati-RSS08} that are popular in POMDP literature to surmount this difficulty.   

%An example of medical diagnosis and treatment was used to illustrate both approaches. We also show the applicability of the approach to a security problem where potentially hostile targets need to be classified before they reach secure or sensitive areas.
%\section{Future work}
%In our current setup, we only consider a single cost function, where in many practical scenarios, there could be multiple cost functions. For example, there could be cost of time and money, respectively. It would be interesting to study how such setting affects our problem solution.

%Furthermore, our current solution is an exact method which completely unfolds the POMDP up to a predefined horizon $D$. The price to pay for this exact solution is the exponential complexity with respect to $D$ to get the unfolded MDP. A possible way to trade off exactness with scalability is to consider point sampling method, which is standard to solve a POMDP in large scale. The difference is that, we may require an bound for the optimality loss, which existing results rarely analyze. 
%\section{Possible Extensions}
%\begin{itemize}
 %   \item Currently only considered the classification problem with single object, where extending to multi objects will be necessary and nontrivial.
  %  \item It is possible to model the classification problem that we are interested in as a two-player partially observable game, where the human to be classified may also wisely avoid to be detected.
%\end{itemize}
\bibliographystyle{IEEEtran}
\bibliography{ref}
\end{document}